\documentclass[a4paper,12pt]{article}
\usepackage[margin=3cm]{geometry}

\usepackage{amsfonts}
\usepackage{amsmath}
\usepackage{amsthm}
\usepackage{amssymb}
\usepackage{tikz}
\usepackage{enumitem}
\usepackage[numbers]{natbib}
\usepackage{dblfloatfix}
\usepackage{etoolbox}
\usetikzlibrary{decorations.pathmorphing}

\newbool{arxiv}
\booltrue{arxiv}

\ifbool{arxiv}{\pretolerance=2100}{}

\pgfdeclarelayer{background}
\pgfdeclarelayer{foreground}
\pgfsetlayers{background,main,foreground}

\newtheorem{theorem}{Theorem}[section]
\newtheorem{lemma}[theorem]{Lemma}
\newtheorem{conjecture}[theorem]{Conjecture}
\theoremstyle{definition}
\newtheorem{defn}[theorem]{Definition}

\newcommand\ignore[1]{}
\newcommand\cat[1]{{\ensuremath{\bf #1}}}
\newcommand\Z{\ensuremath{\mathbb{Z}}}
\newcommand\N{\ensuremath{\mathbb{N}}}

\newcommand\op{\ensuremath{\mathrm{op}}}
\newcommand\Stab{\ensuremath{\mathrm{St}}}
\newcommand\Orbit{\ensuremath{\mathrm{Orbit}}}
\newcommand\Rep{\ensuremath{\mathrm{Rep}}}
\newcommand\Lin{\ensuremath{\mathrm{Lin}}}
\newcommand\Mor{\ensuremath{\mathrm{Mor}}}
\newcommand\Ob{\ensuremath{\mathrm{Ob}}}
\newcommand\Hom{\ensuremath{\mathrm{Hom}}}
\newcommand\id{\ensuremath{\mathrm{id}}}
\newcommand\GA{\cat{GpdAct}{}}
\newcommand\GB{\ensuremath{R_{\cat G}}}
\newcommand\BG{\ensuremath{L _{\cat G}}}
\newcommand\xto[1]{\ensuremath{\smash{{}\stackrel{\smash{#1}}{\to}{}}}}
\newcommand\xtoo[1]{\ensuremath{\smash{{}\stackrel{\smash{#1}}{\Rightarrow}{}}}}
\newcommand{\sdag}{{\ensuremath{\scriptscriptstyle{\dag}}}}
\setitemize{leftmargin=15pt}

\def\fillA{blue!40}
\def\fillB{red!40}

\def\fillD{yellow!50}
\def\fillC{blue!40}
\def\sideangle{30}
\def\nwangle{180-\sideangle}
\def\neangle{\sideangle}
\def\swangle{180+\sideangle}
\def\seangle{-\sideangle}

\def\syntaxfill{blue!20}

\tikzset{circlelabel/.style={draw, circle, inner sep=0pt, minimum width=0.4cm, fill=white}}
\def\innerboxsep{3pt}
\def\innersep{4pt}
\def\sep{5pt}
\newcommand\vc[1]{\begin{tabular}{@{}l}#1\end{tabular}}
\newcommand\separatetwocells{\\[\sep]}
\newcommand{\centerdia}[1]{\ensuremath{#1}}
\newcommand\newtwocell[3]{\centerdia{\begin{aligned}
\begin{tikzpicture}[scale=1.3]
    #1
    \draw [black!20]
        ([xshift=-\innerboxsep, yshift=-\innerboxsep] current bounding box.south west)
        rectangle
        ([xshift=\innerboxsep, yshift=\innerboxsep] current bounding box.north east);
\end{tikzpicture}
\end{aligned}}
& \hspace{15pt} \makebox[120pt][l]{\bf\vc{#2 \\[\innersep] #3}}}

\begin{document}

\allowdisplaybreaks

\title{Groupoid Semantics for Thermal Computing}

\ifbool{arxiv}
{\author{\begin{tabular}{c@{\hspace{1cm}}c}Krzysztof Bar\footnote{Department of Computer Science, University of Oxford}
&
Jamie Vicary\footnote{Centre for Quantum Technologies, National University of Singapore} ${}^* $
\\
krzysztof.bar@cs.ox.ac.uk & jamie.vicary@cs.ox.ac.uk
\end{tabular}}}
{\author{\IEEEauthorblockN{Krzysztof Bar}
\IEEEauthorblockA{Department of Computer Science,
University of Oxford
\\
\texttt{krzysztof.bar@cs.ox.ac.uk}}
\and
\IEEEauthorblockN{Jamie Vicary}
\IEEEauthorblockA{Centre for Quantum Technologies, University of Singapore\\
and Department of Computer Science, University of Oxford
\\
\texttt{jamie.vicary@cs.ox.ac.uk}}
}}

\date{January 14, 2014}

\maketitle

\noindent
\begin{abstract}
A groupoid semantics is presented for systems with both logical and thermal degrees of freedom. We apply this to a syntactic model for encryption, and obtain an algebraic characterization of the heat produced by the encryption function, as predicted by Landauer's principle. Our model has a linear representation theory that reveals an underlying quantum semantics, giving for the first time a functorial classical model for quantum teleportation and other quantum phenomena.
\end{abstract}

\section{Introduction}

\subsection{Overview}
\label{sec:overview}

\noindent
This paper presents a categorical semantics for the thermodynamics of computation. The model describes not only the logical states of a computational device, but also the microscopic states that provide the thermal degrees of freedom. Based on finite groupoids and their actions on sets, the model has a simple combinatorial foundation that gives a new way to model thermodynamic concepts, and is amenable to concrete calculation.

The logical state of a computational device can often be considered to be a coarse-graining of its fundamental microscopic description as a physical system. For example, when reading data from a piece of magnetic tape, all that is important from a logical perspective is the sign of the overall magnetic field created by large number of atoms; variations in the microscopic alignment of particular atoms \mbox{are simply not relevant.}
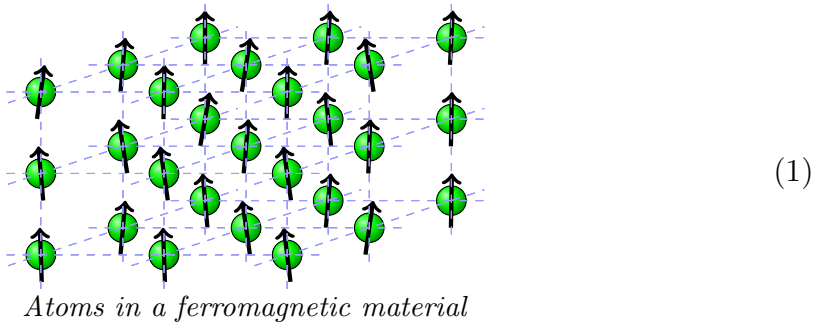
\begin{align}
\begin{aligned}
\def\h{1.2} 
\def\j{1.8} 
\def\k{1.2} 
\def\l{0.4} 
\begin{tikzpicture}[scale=0.9]
\makeatletter
\pgfdeclareradialshading[tikz@ball]{ball}{\pgfqpoint{-10bp}{10bp}}{%
 color(0bp)=(tikz@ball!30!white);
 color(9bp)=(tikz@ball!75!white);
 color(18bp)=(tikz@ball!90!black);
 color(25bp)=(tikz@ball!70!black);
 color(50bp)=(black)}
\makeatother
\pgfmathsetseed{2}
\foreach \x in {0,1,2} {
  \foreach \y in {0,1,2} {
    \foreach \z in {0,1,2} {
      \node [draw, ball color=green, shape=circle] at (\z*\k+\y*\j,\z*\l+\x*\h) {};
      \pgfmathrandominteger{\angle}{0}{20}
      \draw [ultra thick, ->, shift={(\z*\k+\y*\j,\z*\l+\x*\h)}, rotate=\angle-10] (0, -0.4) to (0,0.4);
    }
  }
}
\foreach \x in {0,1,2} {
  \foreach \y in {0,1,2} {
    \foreach \z in {0,1,2} {      
      \draw [dashed, \fillC] (-0.4*\k+\y*\j, -0.4*\l+\x*\h) to (2*\k+0.4*\k+\y*\j,2*\l+0.4*\l+\x*\h);
      \draw [dashed, \fillC] (0+\y*\j+\z*\k, -0.5+\z*\l) to (0+\y*\j+\z*\k, +2*\h+0.5+\z*\l);
      \draw [dashed, \fillC] (-0.5+\z*\k, 0+\x*\h+\z*\l) to (2*\j+0.5+\z*\k, 0+\x*\h+\z*\l);
    }
  }
}
\node at (\j+\k,-0.8) {\small \textit{Atoms in a ferromagnetic material}};
\end{tikzpicture}
\end{aligned}
\end{align}
This is for good reason: while the overall magnetic field orientation is robust, the microscopic alignment of any atom will change uncontrollably over time as it interacts with its neighbours. It would be useless to deliberately encode information into such a local state if one had any desire to recover it later. We can think of logical states as equivalence classes of microscopic states, where information about such non-robust and irrelevant aspects of the state is neglected.

The primary evidence that a thermodynamic interpretation of our model is appropriate arises from a demonstration in Section~\ref{sec:2category} that our model can replicate the basic properties of logical states and microstates that one would expect to find in a toy model of a thermodynamic system. In particular, we show that logical states are robust and amenable to logical operations, while microscopic states are vulnerable to perturbations, and propagate uncontrollably through an extended system.

Our first application is to information-theoretically secure encrypted communication. We develop in Section~\ref{sec:encryption} a syntactic definition of encryption based on complementary observables, and show that this has a representation in our new semantics, in which the encryption operation acts as a permutation on microstates. We further show that this operation has distinct logical and thermal outputs, giving a clear signature of Landauer's principle operating within our formalism.

We then show in Section~\ref{sec:quantum} that our model has a structure-preserving map to an underlying quantum semantics. Applied to our encrypted communication procedure, we show that this yields precisely quantum teleportation. This is significant, providing the first complete account of quantum teleportation in classical combinatorial terms. It  is stronger than previous `toy models' for quantum teleportation described in the literature~\cite{ce08-tqc, cp01-cae, osw05-can, s04-dev, sv13-bsnc} as the correspondence here is functorial in nature.

This reinterpretation suggests a surprisingly close relationship between the properties of information flow in classical thermodynamic systems and in pure-state quantum theory, which does not seem to have been well-recognized in the literature. In particular, the phenomenon of quantum decoherence is seen to correspond to the classical process whereby microscopic perturbations spread through a thermodynamic system. We also give a thermodynamic dense coding algorithm, in which 2 classical bits are apparently sent with perfect fidelity through a channel with a 1-bit classical capacity, and show that under our quantum reinterpretation this gives exactly the known quantum dense coding algorithm. It is exciting to be able to use logical techniques to gain insight into these aspects of quantum theory.

An interesting mathematical feature of our results is that we are able, apparently for the first time, to obtain a combinatorial explanation of linear algebraic phenomena (in this case, teleportation and dense coding)\ that in general involve complex-valued coefficients. This achieves a long-standing goal of the groupoidification programme in categorical algebra, as we discuss in Section~\ref{sec:otherwork}.

\ifbool{arxiv}{}{Throughout the paper we give as much mathematical detail as possible, although some proofs are omitted for reasons of space.}


\subsection{The groupoid model}
\label{sec:groupoidmodel}

\noindent
Logical methods in computer science are traditionally focused on properties of logical states in the sense described above. However, microstates are also of great importance to the foundations of computer science. A key insight is that classical computers can be understood as classical mechanical systems~\cite{fredkintoffoli}, and as such they must be fundamentally time-reversible, with computation proceeding at the microscopic level entirely by  permutations of microstates.

As a result, if we compute a logically irreversible function---such as the addition of two bits modulo 2---we must accept that the initial logical state is unavoidably encoded in the final microstate of the computer. An immediate consequence is the  inevitability of side-channel attacks~\cite{malacaria}, whereby knowledge of the final microstate is used to infer information about the initial logical state. Landauer's principle tells us that such logically irreversible computations necessarily lead to the production of heat~\cite{ladyman-groisman, l61-dhg}.

This article presents a new logical framework in which logical states and microstates can be treated together. Extended systems in our model are described by \textit{groupoids}, which are categories in which all the morphisms are invertible. We interpret the \textit{objects} of a groupoid as describing the logical states, and the \textit{morphisms} as describing the microstates. Here is an example of a groupoid, the union of two copies of the group~$\Z_2$:
\tikzset{circlelabel/.style={draw, circle, inner sep=0pt, minimum width=0.5cm}}
\tikzset{bullet/.style={draw, fill=white, circle, inner sep=0pt, minimum width=0.4cm}}
\def\nodepos{\ifbool{arxiv}{5cm}{3.25}}
\begin{align*}
&\begin{aligned}
\begin{tikzpicture}[yscale=0.9]
\node (S) [bullet] at (-0.25,0) {0};
\node (T) [bullet] at (\nodepos,0) {1};
\draw [->] (S.145) to [out=145, in=-145, looseness=8] node [auto, swap, inner sep=1pt] {$(0,0)$} (S.-145);
\draw [->] (S.35) to [out=35, in=-35, looseness=8] node [auto, inner sep=1pt] {$(1,1)$} (S.-35);
\draw [->] (T.145) to [out=145, in=-145, looseness=8] node [auto, swap, inner sep=1pt] {$(0,1)$} (T.-145);
\draw [->] (T.35) to [out=35, in=-35, looseness=8] node [auto, inner sep=1pt] {$(1,0)$} (T.-35);
\end{tikzpicture}
\end{aligned}
\end{align*}
According to our interpretation this represents a system with 2 logical states, each of which has 2 underlying microstates, giving a total of 4 microstates overall.

In our model, the result of a computation is a morphism in a groupoid. Suppose we perform a deterministic computation on a pair of bits, the outcome of which is the one of the 4 morphisms of the groupoid above, determined by the labelling we have provided. Then  on the level of logical states our computation is performing addition modulo 2. However, no information has been lost in principle, since if we can determine the microstate then we can read off the original states of the two bits.

In this way our model exhibits the essential basic property described above: invertible dynamics at the level of microstates can implement noninvertible computations at the level of logical states. If we assume the thermal interpretation of microstates as mentioned in that section, then the multiple microstates per logical state available in \cat H tell us that implementing this procedure will generate heat.

The algebraic structure of the groupoid plays a crucial role, governing the way that microstates scramble through an extended computational system. From a security perspective, this has the effect in our model of denying an attacker access to the microstate unless she is in possession of the entire extended system, a reasonable property to expect for a toy model of a thermodynamic system. Furthermore, we demonstrate in Section~\ref{sec:phenomenology} that many other simple properties of logical states and microstates hold true in our model, such as the robustness of logical states, and vulnerability of microstates to perturbation. This provides a mathematical foundation for the thermodynamic interpretation we are proposing.

\subsection{Relationships with other work}
\label{sec:otherwork}

\noindent
This paper extends the methods of categorical quantum mechanics~\cite{ac04-csqp, ac08-cqm, cd11-iqo, sv13-bsnc, v12-hsqp} to the thermodynamics of computation. These techniques have already been given broad application, to areas including including linguistics~\cite{csc10-cdm}, discrete quantum models~\cite{ce08-tqc} and classical cryptography~\cite{sv13-bsnc}, and this work provides yet another domain.
More must be done to understand how the wide range of results obtained in the field can be useful in the thermodynamic setting.

The motivations behind this paper, and the formal tools that we use,  have close connections to those of the combinatorial species programme of Joyal~\cite{bicategoryspecies, joyalspecies} and the groupoidification programme of Baez, Dolan and collaborators~\cite{b08-hda7, bd01-fsfd, m06-caqm, m10-ct, mv13-cha1}. That work seeks to find combinatorial models for phenomena in linear algebra, just as we do here. However, the formal basis is slightly different: while these models are based on the bicategory~\cat{Span_2(Gpd)} of groupoids, spans, and spans of spans, our work is based on the bicategory $\GA$ groupoids, profunctors, and spans. These are small technical differences, but they are essential for the results of this paper; in particular, we conjecture that there are no complementary structures in \cat{Span(Gpd)}.

One strong reason to hope for such a connection is our demonstration that a functorial groupoidal interpretation  can be given for linear algebraic phenomena which involve nontrivial complex-valued coefficients, such as quantum teleportation. This has been  a long-stated aim of the groupoidification programme~\cite[Section~1]{b08-hda7}, which has not yet been achieved unless the complex numbers are put in `by hand'~\cite{m06-caqm}.

Groupoids play a central role in homotopy type theory~\cite{hofmannstreicher, htt}, which is a type-theoretical foundation for homotopy theory, and more ambitiously for all of mathematics. Our bicategory \cat{GpdAct} can be seen as a groupoid-based logic with a non-cartesian tensor product, and we conjecture that it may act as model for a linear form of homotopy type theory, although such a linear form has not yet been given in the literature.

There is significant ongoing debate in the philosophy of computation on the status of Landauer's principle~\cite{lps08, lr13-ld, n11-wfl}. A key point of contention is the nature of controlled operations on thermodynamic systems. It is interesting that our own approach also relies heavily on the concept of controlled operation (see Section~\ref{sec:phenomenology}), and in particular is clearly supportive of the argument of Ladyman and Robertson~\cite[Section 4]{lr13-ld}, in that there is no possibility in our model for a system to perform a controlled operation `on itself'. We speculate that logical techniques could, if further developed, play an important role in clarifying these philosophical questions.

\subsection{Acknowledgements}

\noindent
We are grateful to Samson Abramsky, Andrew Ker and Pasquale Malacaria for discussions on the thermodynamics of security, and John Baez, Brendan Fong and Jeffrey Morton for conversations about groupoids and actions.

\section{Groupoids, actions and spans}
\label{sec:2category}

\subsection{Introduction}

\noindent
Our groupoid semantics is presented formally as a bicategory~\cite{borceux, ml97-cwm}. The theory of bicategories has had longstanding and significant application to logical methods in computer science, including in such mainstream areas as the $\lambda$\-calculus~\cite{seely} (presented at one of the first IEEE meetings on Logic in Computer Science), combinatorics~\cite{bicategoryspecies}, concurrency~\cite{presheafpi} and linguistics~\cite{lambek}.

A bicategory has three levels of structure: objects, morphisms and 2-morphisms. For our purposes, a useful interpretation is that the objects describe the types of extended systems, the morphisms describe the types of local systems, and the 2\-morphisms describe the computational processes that can occur. An intuitive description of our bicategory is that objects are groupoids, morphisms are sets equipped with commuting actions of groupoids, and 2\-morphisms are spans of sets compatible with the groupoid actions.

We define our  bicategory $\GA$ in some detail in Section~\ref{fulldefinition}. We go on to describe our abstract syntax in Section~\ref{sec:syntax}, and its representation in $\GA$ in Section~\ref{sec:representation}. In Section~\ref{sec:phenomenology} we show that logical states and microstates described by this bicategory behave in the way we would expect in a simple toy model of a thermodynamic system, such as the ferromagnetic material presented in the introduction.

\subsection{Definition of \GA}
\label{fulldefinition}

\noindent
Here we give a precise definition of the bicategory \GA. We assume some familiarity with the theory of bicategories. For a good introduction see~\mbox{\cite[Chapter 7]{borceux}}.

For any category \cat C, we write $\N : \cat C \to \cat {Set}$ for the constant functor that sends all morphisms to the identity on the set of natural numbers. Also, for any two functors $S,T: \cat C \to \cat {Set}$, we write $S \times T$ for their product in the functor category.

\setitemize{leftmargin=15pt}
\begin{defn}
The symmetric monoidal bicategory $\GA$ is built from the following structures:
\begin{itemize}
\item \textbf{Objects} are locally finite groupoids $\cat G$, $\cat H$, ...
\item \textbf{Morphisms} $S: \cat G \to \cat H$ are functors $S: \cat H ^\op \times \cat G \to \cat {Set}$
\item \textbf{2-Morphisms} $\sigma: S \Rightarrow T$ are natural transformations \mbox{$\sigma: S \times T \Rightarrow \N$}
\end{itemize}
\end{defn}

\noindent
The morphisms of this bicategory are also known as profunctors or distributors, and there is a standard way to compose them~\cite[Proposition 7.8.2]{borceux} which we use here. The 2\-morphisms can be considered as families of spans of sets which are compatible with the groupoid actions. We define 2\-cell composition as follows:
\begin{itemize}
\item \textbf{Vertical composition} is ordinary composition of spans
\item \textbf{Horizontal composition} of 2-morphisms

\vspace{-10pt}
$$\begin{tikzpicture}
\node (1) at (0,0) {\cat G};
\node (2) at (2.5,0) {\cat H};
\node (3) at (5,0) {\cat J};
\draw [->] (1) to [out=20, in=160] node (S) [above] {$S$} (2);
\draw [->] (1) to [out=-20, in=-160] node (T) [below] {$T$} (2);
\draw [->] (2) to [out=20, in=160] node (U) [above] {$U$} (3);
\draw [->] (2) to [out=-20, in=-160] node (V) [below] {$V$} (3);
\node [rotate=-90] (A) at (1.25,0) {$\Rightarrow$};
\node [anchor=west] at (A.center) {$\sigma$};
\node [rotate=-90] (A) at (3.75,0) {$\Rightarrow$};
\node [anchor=west] at (A.center) {$\tau$};
\end{tikzpicture}$$

\vspace{-10pt}
is the 2\-morphism $(\tau \circ \sigma) : U \circ S \Rightarrow V \circ T$ defined as follows, for all $G \in \Ob(\cat G)$, $H,H' \in \Ob(\cat H)$, $J \in \Ob(\cat J)$, $u \in U(J,H)$, $s \in S(H,G)$, $v \in V(J,H')$ and $t \in T(H',G)$:
$$(\tau \circ \sigma) ( (s,u), (t,v) ) := \sum _{H' \stackrel f \to H} \sigma (s \cdot f, t) \times \pi(u, f \cdot v)$$
\end{itemize}
There is a tensor product given on objects by the cartesian product of groupoids. This is the standard tensor product on the bicategory of categories, profunctors and natural transformations, which our bicategory extends.

\begin{defn}
For any 2-morphism $\sigma: S \Rightarrow T$, its \textit{converse} $\sigma ^\dag : T \Rightarrow S$ is constructed as the natural transformation \mbox{$T \times S \Rightarrow S \times T \Rightarrow \N$}, where the first is the symmetry of the cartesian product and the second is the definition of $\sigma$.
\end{defn}

\noindent
If $\sigma$ is interpreted as a computational process, then $\sigma ^\dag$ is interpreted as its time-reversal.

It is not possible in this{\ifbool{arxiv}{ }{ conference }}paper to give a full proof of correctness of the symmetric monoidal bicategory structure. Just writing out the definition of a symmetric monoidal bicategory runs to several pages~\cite[Section 2.2]{sp11-thesis}. The similarity of our construction to the standard symmetric monoidal bicategory of categories, profunctors and natural transformation means that it seems quite a mild assumption that our new bicategory is well-defined.

\ignore{
We do check one detail here, to give a sense of the necessary proof style.
\begin{lemma}
Horizontal composition of 2\-cells in $\GA$ is well-defined.
\end{lemma}
\begin{proof}
LET'S SEE IF WE HAVE SPACE FOR THIS.
\end{proof}
}

\subsection{Graphical syntax}
\label{sec:syntax}

\noindent
We use a graphical syntax to represent the elements of our theory. An application of the graphical notation for bicategories, its topological rewrite rules are closely related to those of topological cobordisms~\cite{sp11-thesis}. This sort of logic would be almost impossible to work with in the absence of categorical tools, which makes the graphical syntax rigorous.

The graphical perspective is the most powerful for abstract reasoning, but it is now well-recognized that there are several other perspectives which are just as valid~\cite{bs10-rosetta, s11-sgl}:
$$\def\gap{6pt}
\begin{tabular}{|@{\,\,}l@{\hspace{\gap}}l@{\hspace{\gap}}l@{\hspace{\gap}}l@{\,\,}|}
\hline
\bf Categories & Objects & Morphisms & 2-Morphisms
\\\hline
\bf Geometry & Regions & Lines & Vertices
\\\hline
\bf Logic & Types & Terms & Rewrites
\\
\hline
\bf Semantics & Groupoids & Actions & Spans
\\
\hline
\end{tabular}
$$
In particular, our graphical approach could in principle be presented in a more standard way as a sequent calculus, although such an approach would be much less convenient.

The basic building blocks of our graphical syntax follow below. Their interpretation in terms of categorical structure is standard~\citep[Section~2.2]{l06-faaa}, and follows the first two lines of the table above: regions represent objects, lines represent morphisms, and vertices represent 2\-morphisms. Horizontal and vertical composition is given by horizontal and vertical juxtaposition respectively. The tensor product of our bicategory is represented by layering one region above another.

A shaded region represents a nontrivial object, represented in our semantics by a groupoid \cat G, while an unshaded region represents the trivial system corresponding to the groupoid~\cat 1:
\begin{align*}
\newtwocell{
    \draw [white] (0.2,1) to (1.8,1);
    \begin{pgfonlayer}{foreground}
    \end{pgfonlayer}
    \draw [fill=\syntaxfill, draw=none] (0.2,0.5)
        to (0.2,1.5)
        to (1.8,1.5)
        to (1.8,0.5)
        to (0.2,0.5);
}
{Nontrivial system \cat G}
{}
\separatetwocells
\newtwocell{
    \draw [white] (0.2,1) to (1.8,1);
    \draw [white] (0.5,0.5) to (0.5,1.5);}
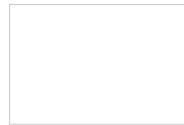
{Trivial system \cat 1}
{}
\end{align*}
Our systems have left and right boundaries, represented in the following way:
\begin{align*}
\newtwocell{
    \draw [white] (0.2,1) to (1.8,1);
    \draw [fill=\syntaxfill, draw=none] (0.2,0.5)
        to (1,0.5) to (1,1.5) to (0.2,1.5) to (0.2,0.5);
    \draw [thick] (1,0.5) to (1,1.5);
}
{Right boundary $\GB$}
{}
\separatetwocells
\newtwocell{
    \draw [white] (0.2,1) to (1.8,1);
    \draw [fill=\syntaxfill, draw=none] (1.8,0.5)
        to (1,0.5) to (1,1.5) to (1.8,1.5) to (1.8,0.5);
    \draw [thick] (1,0.5) to (1,1.5);
}
{Left boundary $\BG$}
{}
\end{align*}
Certain processes can take place which change the connectivity of our systems, represented by the following diagrams, which we interpret as taking place over time, reading from bottom to top:
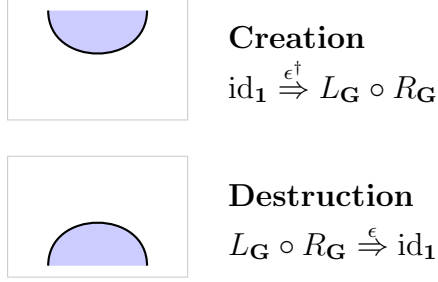
\begin{align*}
\newtwocell{
    \draw [fill=\syntaxfill, draw=none] (0.2,-0.5)
        to (1.8,-0.5)
        to (1.8,-1.5)
        to (0.2,-1.5)
        to (0.2,-0.5);
    \draw [fill=white, thick] (0.5,-0.5)
        to [out=down, in=down, looseness=1.5] (1.5,-0.5);
}
{Separation}
{$\id _{\cat G} \xtoo {\mu ^\dag} \GB \circ \BG$}
\separatetwocells
\newtwocell{
    \draw [fill=\syntaxfill, draw=none] (0.2,0.5)
        to (1.8,0.5)
        to (1.8,1.5)
        to (0.2,1.5)
        to (0.2,0.5);
    \draw [fill=white, thick] (0.5,0.5)
        to [out=up, in=up, looseness=1.5] (1.5,0.5);
}
{Gluing}
{$\GB \circ \BG \xtoo {\mu} \id _{\cat G}$}
\separatetwocells
\newtwocell{
    \draw [white] (0.2,1) to (1.8,1);
    \draw [fill=\syntaxfill, thick] (0.5,1.5)
        to [out=down, in=down, looseness=1.5] (1.5,1.5);
    \draw [thick, white] (1,0.5) to (1,1);
}
{Creation}
{$\id _{\cat 1} \xtoo {\epsilon ^\sdag} \BG \circ \GB$}
\separatetwocells
\newtwocell{
    \draw [white] (0.2,-1) to (1.8,-1);
    \draw [fill=\syntaxfill, thick] (0.5,-1.5)
        to [out=up, in=up, looseness=1.5] (1.5,-1.5);
    \draw [thick, white] (1,-0.5) to (1,-1);
}
{Destruction}
{$\BG \circ \GB \xtoo \epsilon \id _{\cat 1}$}
\end{align*}

\noindent
We impose the following equivalences between composite operations, which correspond mathematically to the ambidextrous adjointness conditions~\citep{l06-faaa}:
\def\innerboxsep{8pt}
\def\compactwidthscale{0.55}
\def\compactheightscale{0.7}
\def\drawspecialboxA{    \draw [black!20]
        ([xshift=-\compactheightscale*\innerboxsep, yshift=-\compactwidthscale*\innerboxsep] -3,0)
        rectangle
        ([xshift=\compactheightscale*\innerboxsep, yshift=\compactwidthscale*\innerboxsep] 1,2);}
\def\drawspecialboxB{    \draw [black!20]
        ([xshift=-\compactheightscale*\innerboxsep, yshift=-\compactwidthscale*\innerboxsep] -1,0)
        rectangle
        ([xshift=\compactheightscale*\innerboxsep, yshift=\compactwidthscale*\innerboxsep] 1,2);}
\begin{gather}
\label{eq:top1}
\begin{aligned}
\begin{tikzpicture}[xscale=\compactwidthscale, yscale=\compactheightscale]
\node (a) at (0,0) [inner sep=0pt] {};
\node (b) at (0,1) [inner sep=0pt] {};
\node (c) at (-1,1) {};
\node (d) at (-2,1) {};
\node (e) at (-2,2) [inner sep=0pt] {};
\draw [fill=\syntaxfill, draw=none] (a.center)
    to (b.center)
    to [out=up, in=up, looseness=1.5] (c.center)
    to [out=down, in=down, looseness=1.5] (d.center)
    to (e.center)
    to (-3,2) to (-3,0) to (0,0);
\draw [thick] (a.center)
    to (b.center)
    to [out=up, in=up, looseness=1.5] (c.center)
    to [out=down, in=down, looseness=1.5] (d.center)
    to (e.center);
\draw [white] (1,0) to (1,1);
\drawspecialboxA
\end{tikzpicture}
\end{aligned}
\hspace{4pt}
=
\hspace{4pt}
\begin{aligned}
\begin{tikzpicture}[xscale=\compactwidthscale, yscale=\compactheightscale]
\node (a) at (0,0) {};
\node (e) at (0,2) {};
\draw [fill=\syntaxfill, draw=none] (a.center)
    to (e.center)
    to (-1,2) to (-1,0) to (0,0);
\draw [thick] (a.center)
    to (e.center);
\draw [white] (1,0) to (1,1);
\drawspecialboxB
\end{tikzpicture}
\end{aligned}
\hspace{4pt}
=
\hspace{4pt}
\begin{aligned}
\begin{tikzpicture}[xscale=\compactwidthscale, yscale=-\compactheightscale]
\node (a) at (0,0) {};
\node (b) at (0,1) {};
\node (c) at (-1,1) {};
\node (d) at (-2,1) {};
\node (e) at (-2,2) {};
\draw [fill=\syntaxfill, draw=none] (a.center)
    to (b.center)
    to [out=up, in=up, looseness=1.5] (c.center)
    to [out=down, in=down, looseness=1.5] (d.center)
    to (e.center)
    to (-3,2) to (-3,0) to (0,0);
\draw [thick] (a.center)
    to (b.center)
    to [out=up, in=up, looseness=1.5] (c.center)
    to [out=down, in=down, looseness=1.5] (d.center)
    to (e.center);
\draw [white] (1,0) to (1,1);
\drawspecialboxA
\end{tikzpicture}
\end{aligned}
\\
\label{eq:top2}
\begin{aligned}
\begin{tikzpicture}[xscale=-\compactwidthscale, yscale=\compactheightscale]
\node (a) at (0,0) {};
\node (b) at (0,1) {};
\node (c) at (-1,1) {};
\node (d) at (-2,1) {};
\node (e) at (-2,2) {};
\draw [fill=\syntaxfill, draw=none] (a.center)
    to (b.center)
    to [out=up, in=up, looseness=1.5] (c.center)
    to [out=down, in=down, looseness=1.5] (d.center)
    to (e.center)
    to (-3,2) to (-3,0) to (0,0);
\draw [thick] (a.center)
    to (b.center)
    to [out=up, in=up, looseness=1.5] (c.center)
    to [out=down, in=down, looseness=1.5] (d.center)
    to (e.center);
\draw [white] (1,0) to (1,1);
\drawspecialboxA
\end{tikzpicture}
\end{aligned}
\hspace{4pt}
=
\hspace{4pt}
\begin{aligned}
\begin{tikzpicture}[xscale=-\compactwidthscale, yscale=\compactheightscale]
\node (a) at (0,0) {};
\node (e) at (0,2) {};
\draw [fill=\syntaxfill, draw=none] (a.center)
    to (e.center)
    to (-1,2) to (-1,0) to (0,0);
\draw [thick] (a.center)
    to (e.center);
\draw [white] (1,0) to (1,1);
\drawspecialboxB
\end{tikzpicture}
\end{aligned}
\hspace{4pt}
=
\hspace{4pt}
\begin{aligned}
\begin{tikzpicture}[xscale=-\compactwidthscale, yscale=-\compactheightscale]
\node (a) at (0,0) {};
\node (b) at (0,1) {};
\node (c) at (-1,1) {};
\node (d) at (-2,1) {};
\node (e) at (-2,2) {};
\draw [fill=\syntaxfill, draw=none] (a.center)
    to (b.center)
    to [out=up, in=up, looseness=1.5] (c.center)
    to [out=down, in=down, looseness=1.5] (d.center)
    to (e.center)
    to (-3,2) to (-3,0) to (0,0);
\draw [thick] (a.center)
    to (b.center)
    to [out=up, in=up, looseness=1.5] (c.center)
    to [out=down, in=down, looseness=1.5] (d.center)
    to (e.center);
\draw [white] (1,0) to (1,1);
\drawspecialboxA
\end{tikzpicture}
\end{aligned}
\end{gather}
These axioms say that we can topologically deform any boundary without changing the process that it describes.

\subsection{Representing the syntax}
\label{sec:representation}

\noindent
We now show how the syntax is represented in our bicategorical semantics \GA. We simplify the discussion by restricting to \textit{skeletal} groupoids, which are groupoids in which isomorphic objects are equal. This is without loss of generality since every groupoid is equivalent to a skeletal one in our bicategory.

We first describe construction of the boundary morphisms from our groupoid structure.
\begin{defn}
The boundary morphisms are defined in the following way, for all $G \in \Ob(\cat G)$ and $G \xto g G$:
\begin{align*}
\BG (G \xto g G, \bullet) &\mapsto \Hom _{\cat G} (g, G)
\\
\GB (\bullet, G \xto g G) &\mapsto \Hom _{\cat G} (G, g)
\end{align*}
\end{defn}

\noindent
Intuitively the boundary morphisms are built from the set of morphisms of the skeletal groupoid, with action given by composition on the left or the right.

We next give the representation of the 4 basic processes in the theory. Here the symbol $\omega _{a,b}$ is defined to equal 1 if $a=b$, and 0 otherwise.
\begin{defn}
The separation, gluing, destruction and creation processes are defined as the following spans, for all $G \in \Ob(\cat G)$ and $g,g',g'': G \to G$:
\begin{align*}
\mu ^\dag (g, (g', g'')) &= \mu ((g', g''), g) = \omega _{g, g'g''}
\\
\epsilon((g,g'), \bullet) &=\epsilon ^\dag (\bullet, (g,g')) = \omega _{\id_G, gg'}
\end{align*}
\end{defn}

\noindent
These also have simple intuitive descriptions:

\begin{itemize}
\item
{\bf Separation $\mu ^\dag$.} When one system is separated into two, the logical state $G \in \Ob(\cat G)$ is maintained, and the microstate $g: G \to G$ is split into a nondeterministic sum over all pairs of microstates $g,g': G \to G$ such that $gg' = g$.
\item {\bf Gluing $\mu$.}
The gluing process is the converse of the separation process: if the logical states of the two initial systems match, the microstates are multiplied using the groupoid algebra; otherwise the gluing process fails.
\item {\bf Creation $\epsilon ^\dag$.}
The system is created nondeterministically in an arbitrary  logical state $G \in \Ob(\cat G)$, with only the identity microstate $\id _G$ occupied.
\item {\bf Destruction $\epsilon$.} If the system is in an identity microstate, then the system is removed, otherwise the process fails.
\end{itemize}
Given our physical interpretation, the creation process can be thought of as creating our system in a \textit{ground state}.

Some details must be checked to ensure that the definitions above are valid.
\begin{lemma}
The representations of $\mu$, $\mu ^\dag$, $\epsilon$ and $\epsilon ^\dag$ are well-defined, and satisfy the required naturality conditions.
\end{lemma}
\begin{proof}
For $\mu$, the naturality condition is that for all objects $G \in \Ob ( \cat G )$ and all $g_0, g_1,g, g', g'' : G \to G$ we have
\begin{align*}
&& \mu ((g_0 g' , g'' g_1), g_0g g_1) &= \mu ((g', g''), g )
\\
\Leftrightarrow \hspace{-0.8cm}
&& \omega _{g_0 g' g'' g_1, g_0 g g_1} & = \omega _{g'g'', g},
\end{align*}
which is clearly satisfied. The proof for $\mu ^\dag$ is similar. For $\epsilon$ and $\epsilon ^\dag$ the naturality conditions are empty.
\end{proof}
\begin{lemma}
The representations of $\mu$, $\mu ^\sdag$, $\epsilon$ and $\epsilon ^\dag$ satisfy the topological conditions~(\ref{eq:top1}--\ref{eq:top2}).
\end{lemma}
\begin{proof}
This follows from the groupoid axioms. The first diagram of \eqref{eq:top1} acts on some $g:G \to G$ for $G \in \Ob(\cat G)$ the the following way:
$$\textstyle g \stackrel {\,\,\mu ^\dag} \mapsto \sum _{g':G \to G} (g' {}^{-1},g', g) \stackrel \epsilon \mapsto \sum _{g':G \to G} \omega _{\id_G, g'g} \, g' {}^{-1} = g$$
This verifies the first equality of \eqref{eq:top1}. The third diagram of that expression acts on some $g:G \to G$ as follows:
$$g \stackrel {\,\epsilon ^\dag} \mapsto (g, \id_G) \stackrel \mu \mapsto g \, \id_G = g$$
This verifies the second
equality of \eqref{eq:top1}. The equalities of~\eqref{eq:top2} can be proven in a similar way.
\end{proof}

\subsection{Phenomenology}
\label{sec:phenomenology}

\noindent
We now consider the behaviour of logical states and microstates in our model. We begin by analyzing the case when we have an entire piece of computational material under our control, represented by the following diagram:
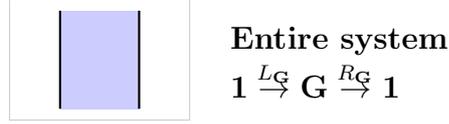
\def\innerboxsep{3pt}
\begin{align*}
\newtwocell{
    \draw [white] (0.2,1) to (1.8,1);
    \draw [fill=\syntaxfill, draw=none] (0.6,0.5)
        to (1.4,0.5) to (1.4,1.5) to (0.6,1.5) to (0.6,0.5);
    \draw [thick] (1.4,0.5) to (1.4,1.5);
    \draw [thick] (0.6,0.5) to (0.6,1.5);
}
{Entire system}
{$\cat 1 \xto \BG \cat G \xto \GB \cat 1$}
\end{align*}
This is an endomorphism of the trivial groupoid \cat 1, and so corresponds to a finite set in our model.
\begin{lemma}
\label{lem:microstates}
For a skeletal groupoid \cat G we have $$\BG \circ \GB \simeq \Mor(\cat G).$$
\end{lemma}
\begin{proof}
\ifbool{arxiv}{Omitted.}{Omitted due to space constraints.}
\end{proof}

\noindent
This has important consequences for our thermodynamic interpretation: if the entire system can be controlled at once, then the set of accessible states is the set of morphisms of the skeletal groupoid, which for us corresponds to the complete set of microstates.

However, we now examine the case that only a part of an extended system can be accessed. We allow in general for it to interact with a second free system with some set $S$ of states. We call this a \textit{controlled operation}, and describe it geometrically in the following way:
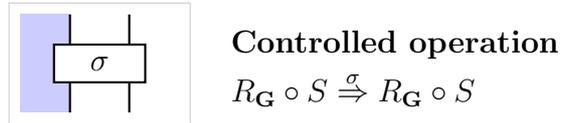
\def\innerboxsep{3pt}
\begin{align*}
\newtwocell{
    \draw [thick, white] (0.2,0.5) to (1.8,0.5);
    \draw [fill=\syntaxfill, draw=none] (0.2,0.5)
        to (0.2,1.5) to (0.7,1.5) to (0.7,0.5) to (0.7,0.5);
    \draw [thick] (0.7,0.5) to (0.7,1.5);
    \draw [thick] (1.3,0.5) to (1.3,1.5);
    \node [draw, thick, minimum width=1.2cm, fill=white, minimum height=0.5cm] at (1,1) {$\sigma$};
}
{Controlled operation}
{$\GB \circ S \xtoo \sigma \GB \circ S$}
\end{align*}
Here $\cat G \xto \GB \cat 1$ represents part of the boundary of the extended system, and $\cat 1 \xto S \cat 1$ represents the free system. We now classify the 2\-morphisms of this type in our model.
\begin{lemma}
\label{lem:controlled}
Given a skeletal groupoid \cat G and a set $S$, a local operation $\GB \circ S \xtoo \sigma \GB \circ S$ is determined by spans of sets $S \to \Hom_\cat G (G,G) \times S$ for each $G \in \Ob(\cat G)$.
\end{lemma} 
\begin{proof}
There is a bijective correspondence between morphisms of the type $\GB \circ S \xtoo \sigma \GB \circ S$ and morphisms of the type $S \xtoo {\sigma'} (\BG \circ \GB) \circ S$, which we construct in the following way:
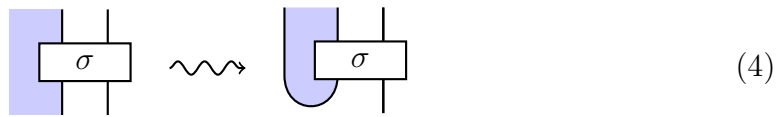
\begin{align}
\begin{aligned}
\begin{tikzpicture} [thick]
    \draw [fill=\syntaxfill, draw=none] (0.0,0.3)
        to (0.0,1.7) to (0.7,1.7) to (0.7,0.3);
    \draw [thick] (0.7,0.3) to (0.7,1.7);
    \draw [thick] (1.3,0.3) to (1.3,1.7);
    \node [draw, thick, minimum width=1.2cm, fill=white, minimum height=0.5cm] at (1,1) {$\sigma$};
\end{tikzpicture}
\end{aligned}
\quad
\begin{aligned}
\begin{tikzpicture}[thick, decoration=snake]
\draw[->,decorate] (0,0) to (1,0);
\end{tikzpicture}
\end{aligned}
\quad
\begin{aligned}
\begin{tikzpicture} [thick]
    \draw [fill=\syntaxfill, draw=none] (0, 1.7)
    to (0, 0.8) 
    to [out=down, in=down, looseness=2](0.7,0.8) to (0.7,1.7);
    \draw [thick] (1.3,0.3) to (1.3,1.7);
    \draw [thick] (0, 1.7)
    to (0, 0.8) 
    to [out=down, in=down, looseness=2](0.7,0.8) to (0.7,1.7);
    \draw [thick] (1.3,0.3) to (1.3,1.7);
    \node [draw, thick, minimum width=1.2cm, fill=white, minimum height=0.5cm] at (1,1) {$\sigma$};
\end{tikzpicture}
\end{aligned}
\end{align}
\begin{align}
\qquad
\begin{aligned}
\begin{tikzpicture} [thick]
    \draw [fill=\syntaxfill, draw=none] (0.5,1)
        to (0.5,1.9) to (0.9,1.9) to (0.9,1);
    \draw [thick] (0.9,1) to (0.9,1.9);
        \draw [thick] (0.5,1) to (0.5,1.9);
    \draw [thick] (1.4,0.3) to (1.4,1.9);
    \node [draw, thick, minimum width=1.2cm, fill=white, minimum height=0.5cm] at (1,1) {$\sigma'$};
\end{tikzpicture}
\end{aligned}
\quad
\begin{aligned}
\begin{tikzpicture}[thick, decoration=snake]
\draw[->,decorate] (0,0) to (1,0);
\end{tikzpicture}
\end{aligned}
\quad
\begin{aligned}
\begin{tikzpicture} [thick]
    \draw [fill=\syntaxfill, draw=none] (0.9,1.9) to (0.9,1.2) 
        to(0.5,1.2) 
        to [out=up, in=up, looseness=2](-0.2,1.2)
        to (-0.2, 0.3)
        to (-0.6, 0.3)
        to (-0.6, 1.9);
    \draw [thick] (0.9,1) to (0.9,1.9);
        \draw [thick] (0.5,1.2) 
        to [out=up, in=up, looseness=2](-0.2,1.2)
        to (-0.2, 0.3);
    \draw [thick] (1.4,0.3) to (1.4,1.9);
    \node [draw, thick, minimum width=1.2cm, fill=white, minimum height=0.5cm] at (1,1) {$\sigma'$};
\end{tikzpicture}
\end{aligned}
\label{eq:twistings'}
\end{align}
That these constructions are inverse follows from the topological properties of the calculus. By Lemma~\ref{lem:microstates} the type of $\sigma'$ is simply a span of sets $S \to \Mor(\cat G) \times S$, and the conclusion follows since $\Mor(\cat G) \simeq \bigcup _G \Hom _\cat G(G,G)$.
\end{proof}

\noindent
By considering the form of transformation~\eqref{eq:twistings'}, we see that $\sigma$ multiplies the microstate by some group element, and thus that the logical state cannot be changed.

We can describe the consequences of Lemma~\ref{lem:controlled} as follows: a controlled operation constitutes, for each logical state $G$ of the extended system \cat G, and for each state of the free system $S$, a rule that multiplies the underlying microstate of $G$ by some arbitrary group element, and also arbitrarily updates the state of the free system $S$.

We can use this to infer certain properties of logical states and microstates under the action of local operations:
\begin{enumerate}[label=(\arabic*)]
\item \textbf{Logical states are robust against  local perturbations} since there is provision for changing the logical state.

\item \textbf{Microstates are not robust against local perturbations} since it is clear that the microstate can change.

\item \textbf{Logical states can control conditional operations} since we can allow different behaviour depending on the logical state.

\item \textbf{Microstates cannot control conditional operations} since there is no provision for different behaviour depending on the initial microstate.

\item \textbf{Logical states can be ascertained by local operations} since we could choose to set $S$ to a different state for each logical state of the extended system.

\item \textbf{Microstates cannot be ascertained by local operations} since there is no provision for different behaviour depending on the initial microstate.
\end{enumerate}

\noindent
This analysis gives us a lot of information about how logical states and microstates behave. Since these are reasonable properties of a toy model of a thermodynamic system, we can feel confident that our thermodynamic interpretation is sound.

\tikzset{circlelabel/.style={draw, circle, fill=white, fill opacity=1, thick}}

\section{Modelling Encryption}
\label{sec:encryption}

\subsection{Introduction}

\noindent
The aim of this section is to develop a model of encrypted communication in our bicategory \cat{GpdAct}, which we achieve in Section~\ref{sec:concreteencryption}. We find that the encryption procedure is many-to-one on logical states, as required for security, yet invertible at the level of microstates, which is consistent with the discussion in Sections~\ref{sec:overview} and~\ref{sec:groupoidmodel}.

A key technical  contribution is made in Section~\ref{sec:bigproof}, where we prove syntactically that communication structures can be built from complementary structures. A similar result has been given previously in the monoidal category context~\cite{cd11-iqo}, although that proof assumes a strong extra property that we do not require here.

The results on this section rely on the existing bicategorical definitions of complementary structures and communication structures~\cite{v12-hqt,v12-hsqp}. It was shown in those papers that using a quantum semantics, communication structures correspond to quantum teleportation procedures, and complementary structures correspond to complementary observables, also called mutually-unbiased bases. In Section~\ref{sec:concretecomplementarity} we show that they also have models in our new categorical semantics.

We repeat the definitions of complementary and communication structures here, since they are central to our results. We use a slightly different form than given in the literature that avoids the presence of scalar factors. Recall that these diagrams describe processes taking place over time, where in our notation time runs from bottom to top. We assume throughout that our topological axioms~(\ref{eq:top1}--\ref{eq:top2}) are satisfied.
\begin{defn}
\label{def:complementarystructure}
A \textit{complementary structure} is a 2\-morphism $\delta$ of the following type, such that both $\delta$ and its partial transpose are unitary:

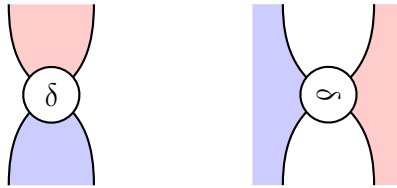
\begin{equation*}
\begin{aligned}
\begin{tikzpicture} [thick, scale=0.8, fill opacity=0.5]
\draw [fill=\fillB, draw=none, fill opacity=0.5] (0.7,0)
to [out=down, in=\neangle](0,-1.5)
to [out=\nwangle, in=down](-0.7, 0);
\draw  (0.7,0)
to [out=down, in=\neangle](0,-1.5)
to [out=\nwangle, in=down](-0.7, 0);
\draw [fill=\fillA, draw=none] (0.7,-3)
to [out=up, in=\seangle](0,-1.5)
to [out=\swangle, in=up](-0.7, -3);
\draw (0.7,-3)
to [out=up, in=\seangle](0,-1.5)
to [out=\swangle, in=up](-0.7, -3);
\node [circlelabel, rotate=10]
 at (0, -1.5) {$\delta$};
\end{tikzpicture}
\end{aligned}
\hspace{2cm}
\begin{aligned}
\begin{tikzpicture} [thick, scale=0.4, fill opacity=0.5]
\draw [fill=\fillA, draw=none, fill opacity=0.5] (-2.5, 3)
to (-1.5,3)
to [out=down, in=\nwangle](0,0)
to [out=\swangle, in=up](-1.5,-3)
to (-2.5,-3);
\draw (-1.5,3)
to [out=down, in=\nwangle](0,0)
to [out=\swangle, in=up](-1.5,-3);
\draw [fill=\fillB, draw=none] (2.5, 3)
to (1.5,3)
to [out=down, in=\neangle](0,0)
to [out=\seangle, in=up](1.5,-3)
to (2.5,-3);
\draw (1.5,3)
to [out=down, in=\neangle](0,0)
to [out=\seangle, in=up](1.5,-3);
\node [circlelabel, rotate=-80]
 at (0, 0) {$\delta$};
\end{tikzpicture}
\end{aligned}
\end{equation*}
\end{defn}
\begin{defn}
A \textit{communication procedure} is a 2\-morphism $\lambda$ of the following type, such that both $\lambda$ and its partial transpose are unitary:

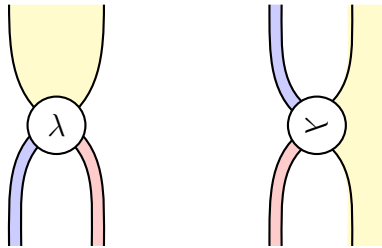
\begin{equation*}
\begin{aligned}
\begin{tikzpicture} [thick, scale=0.8, fill opacity=0.5]
\draw [fill=\fillD, draw=none, fill opacity=0.5](0.78,2)
to [out=down, in=\neangle](0,0)
to [out=\nwangle, in=down](-0.78, 2);
\draw (0.78,2)
to [out=down, in=\neangle](0,0)
to [out=\nwangle, in=down](-0.78, 2);
\draw [fill=\fillB, draw=none] (0.78,-2)
to [out=up, in=\seangle](0.1,0.1)
to (0, -0.1)
to [out=\seangle, in=up](0.58, -2);
\draw (0.78,-2)
to [out=up, in=\seangle](0.1,0.1)
to (0, -0.1)
to [out=\seangle, in=up](0.58, -2);
\draw [fill=\fillA, draw=none] (-0.78,-2)
to [out=up, in=\swangle](-0.1,0.1)
to (0, -0.1)
to [out=\swangle, in=up](-0.58, -2);
\draw (-0.78,-2)
to [out=up, in=\swangle](-0.1,0.1)
to (0, -0.1)
to [out=\swangle, in=up](-0.58, -2);
\node [draw, fill=white, shape=circle, fill opacity=1]
 at (0, 0) {$\lambda$};
\end{tikzpicture}
\end{aligned}
\hspace{2cm}
\begin{aligned}
\begin{tikzpicture} [thick, scale=0.8, fill opacity=0.5]
\draw [fill=\fillD, draw=none, fill opacity=0.5]
(1.2,-2)to(0.5,-2)
to [out=up, in=\seangle](0,-0.2)
to (0, 0.2)
to [out=\neangle, in=down](0.5,2)
to (1.2,2);
\draw (0.58,-2)
to [out=up, in=\seangle](0,-0.2)
to (0, 0.2)
to [out=\neangle, in=down](0.58,2);
\draw [fill=\fillB, draw=none] (-0.78,-2)
to [out=up, in=\swangle](-0.1,-0)
to (0, -0.2)
to [out=\swangle, in=up](-0.58, -2);
\draw (-0.78,-2)
to [out=up, in=\swangle](-0.1,0)
to (0, -0.2)
to [out=\swangle, in=up](-0.58, -2);
\draw [fill=\fillA, draw=none] (-0.78,2)
to [out=down, in=\nwangle](-0.1,0)
to (0, 0.2)
to [out=\nwangle, in=down](-0.58, 2);
\draw (-0.78,2)
to [out=down, in=\nwangle](-0.1,0)
to (0, 0.2)
to [out=\nwangle, in=down](-0.58, 2);
\node [draw, fill=white, fill opacity=1,shape=circle, rotate=-90]
 at (0, 0) {$\lambda$};
\end{tikzpicture}
\end{aligned}
\end{equation*}
\end{defn}

\subsection{Communication from complementarity}
\label{sec:bigproof}

\noindent
We now see how communication structures and complementary structures are related. Recall from Section~\ref{sec:syntax} that overlapping regions denote the tensor product of the bicategory.

\begin{theorem}
\label{thm:commfromcomp}
We can build a communication structure from a complementary structure in the following way:
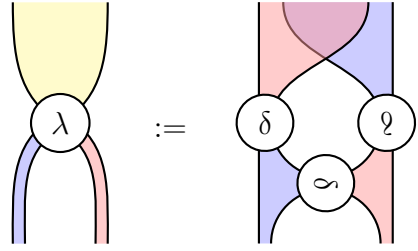
\begin{equation*}
\begin{aligned}
\begin{tikzpicture} [thick, scale=0.8, fill opacity=0.5]
\draw [fill=\fillD, draw=none, fill opacity=0.5](0.78,2)
to [out=down, in=\neangle](0,0)
to [out=\nwangle, in=down](-0.78, 2);
\draw (0.78,2)
to [out=down, in=\neangle](0,0)
to [out=\nwangle, in=down](-0.78, 2);
\draw [fill=\fillB, draw=none] (0.78,-2)
to [out=up, in=\seangle](0.1,0.1)
to (0, -0.1)
to [out=\seangle, in=up](0.58, -2);
\draw (0.78,-2)
to [out=up, in=\seangle](0.1,0.1)
to (0, -0.1)
to [out=\seangle, in=up](0.58, -2);
\draw [fill=\fillA, draw=none] (-0.78,-2)
to [out=up, in=\swangle](-0.1,0.1)
to (0, -0.1)
to [out=\swangle, in=up](-0.58, -2);
\draw (-0.78,-2)
to [out=up, in=\swangle](-0.1,0.1)
to (0, -0.1)
to [out=\swangle, in=up](-0.58, -2);
\node [draw, fill=white, shape=circle, fill opacity=1]
 at (0, 0) {$\lambda$};
\end{tikzpicture}
\end{aligned}
\quad:=\quad
\begin{aligned}
\begin{tikzpicture} [thick, scale=0.8, fill opacity=0.5]
\draw [fill=\fillA, draw=none, fill opacity=0.5]
(-0.9, -4.5)
to [out=up, in=left](0,-3.6)
to (0, -3.4)
to [out=left, in=down](-0.9, -2.5)
to (-1.1, -2.5)
to (-1.1, -4.5);
\draw (-0.9, -4.5)
to [out=up, in=left](0,-3.6)
to (0, -3.4)
to [out=left, in=down](-0.9, -2.5)
to (-1.1, -2.5)
to (-1.1, -4.5);
\draw [fill=\fillA, draw=none]
(1.1, -0.5)
to (1.1, -2.5)
to (0.9, -2.5)
to [out=up, in=down](-0.7, -0.5);
\draw (1.1, -0.5)
to (1.1, -2.5)
to (0.9, -2.5)
to [out=up, in=down](-0.7, -0.5);
\draw [fill=\fillB, draw=none]
(-1.1, -0.5)
to (-1.1, -2.5)
to (-0.9, -2.5)
to [out=up, in=down](0.7, -0.5);
\draw (-1.1, -0.5)
to (-1.1, -2.5)
to (-0.9, -2.5)
to [out=up, in=down](0.7, -0.5);
\draw [fill=\fillB, draw=none]
(0.9, -4.5)
to [out=up, in=right](0,-3.6)
to (0, -3.4)
to [out=right, in=down](0.9, -2.5)
to (1.1, -2.5)
to (1.1, -4.5);
\draw (0.9, -4.5)
to [out=up, in=right](0,-3.6)
to (0, -3.4)
to [out=right, in=down](0.9, -2.5)
to (1.1, -2.5)
to (1.1, -4.5);
\node [draw, fill=white, fill opacity=1, shape=circle, rotate=-80]
 at (0, -3.5) {$\delta$};
 \node [draw, fill=white, fill opacity=1, shape=circle, rotate=10]
 at (-1, -2.5) {$\delta$};
  \node [draw, fill=white, fill opacity=1, shape=circle, yscale=-1, rotate=10]
 at (1, -2.5) {$\delta$};
\end{tikzpicture}
\end{aligned}
\end{equation*}
\end{theorem}
\begin{proof}
We must show that both $\lambda$ and its partial transpose $\lambda'$ are unitary. For spans of sets $S \xto \phi T$ and $T \xto \psi S$, then $\phi \cdot \psi = \id \Leftrightarrow \psi \cdot \phi = \id$, and so we need only check $\lambda ^\dag \circ\lambda = \id$ and $\lambda' {}^\dag \circ \lambda' = \id$. We prove the first of these in equation~\eqref{eq:bigproof}, using the definition of $\lambda$, topological properties of the calculus, unitarity of $\delta$, and unitarity of the partial transpose of $\delta$\ifbool{arxiv}{:}{.}
\ifbool{arxiv}{\begin{equation}}
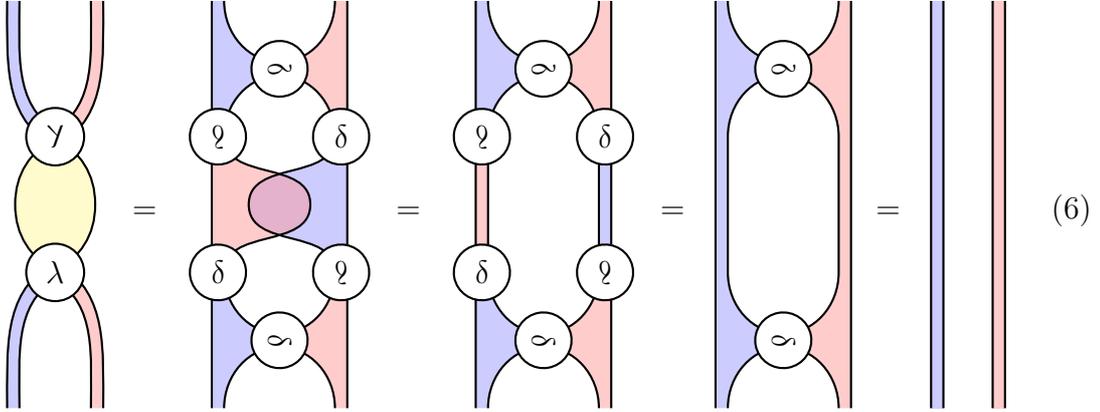
{\begin{figure*}[t!]\begin{align}}
\ifbool{arxiv}{\def\ws{\hspace{0.2cm}}}{\quad}
\ifbool{arxiv}{\tikzset{every picture/.style={xscale=0.9}}}{}
\label{eq:bigproof}
\begin{aligned}
\begin{tikzpicture} [thick, scale=0.9, fill opacity=0.5]
\draw [fill=\fillD, draw=none, fill opacity=0.5] (0.65,0)
to [out=down, in=\neangle](0,-1)
to [out=\nwangle, in=down](-0.65, 0)
to [out=up, in=\swangle](0,1)
to [out=\seangle, in=up](0.65, 0);
\draw (0.65,0)
to [out=down, in=\neangle](0,-1)
to [out=\nwangle, in=down](-0.65, 0)
to [out=up, in=\swangle](0,1)
to [out=\seangle, in=up](0.65, 0);
\draw [fill=\fillB, draw=none] (0.78,-3)
to [out=up, in=\seangle](0.1,-0.9)
to (0, -1.1)
to [out=\seangle, in=up](0.58, -3);
\draw (0.78,-3)
to [out=up, in=\seangle](0.1,-0.9)
to (0, -1.1)
to [out=\seangle, in=up](0.58, -3);
\draw [fill=\fillA, draw=none] (-0.78,-3)
to [out=up, in=\swangle](-0.1,-0.9)
to (0, -1.1)
to [out=\swangle, in=up](-0.58, -3);
\draw  (-0.78,-3)
to [out=up, in=\swangle](-0.1,-0.9)
to (0, -1.1)
to [out=\swangle, in=up](-0.58, -3); 
\draw [fill=\fillB, draw=none] (0.78,3)
to [out=down, in=\neangle](0.1,0.9)
to (0, 1.1)
to [out=\neangle, in=down](0.58, 3);
\draw (0.78,3)
to [out=down, in=\neangle](0.1,0.9)
to (0, 1.1)
to [out=\neangle, in=down](0.58, 3);
\draw [fill=\fillA, draw=none] (-0.78,3)
to [out=down, in=\nwangle](-0.1,0.9)
to (0, 1.1)
to [out=\nwangle, in=down](-0.58, 3);
\draw (-0.78,3)
to [out=down, in=\nwangle](-0.1,0.9)
to (0, 1.1)
to [out=\nwangle, in=down](-0.58, 3);
\node [draw, fill=white, fill opacity=1,shape=circle]
 at (0, -1) {$\lambda$};
\node [draw, fill=white, fill opacity=1, shape=circle, yscale=-1]
 at (0, 1) {$\lambda$}; 
\end{tikzpicture}
\end{aligned}
\ws=\ws
\begin{aligned}
\begin{tikzpicture}[thick, scale=0.9, fill opacity=0.5]
\draw [fill=\fillA, draw=none, fill opacity=0.5] (-0.9,-3)
to (-0.9, -3)
to [out=up, in=left](0,-2.1)
to (0, -1.9)
to [out=left, in=down](-0.9, -1)
to (-1.1, -1)
to (-1.1, -3);
\draw [fill=\fillA, draw=none]
(1.1, -1)
to (0.9, -1)
to [out=up, in=down](-0.5, 0)
to [out=up, in=down] (0.9, 1)
to (1.1, 1);
\draw [fill=\fillB, draw=none]
(-1.1, -1)
to (-0.9, -1)
to [out=up, in=down](0.5, 0)
to [out=up, in=down] (-0.9, 1)
to (-1.1, 1);
\draw [fill=\fillB, draw=none] (0.9,-3)
to (0.9, -3)
to [out=up, in=right](0,-2.1)
to (0, -1.9)
to [out=right, in=down](0.9, -1)
to (1.1, -1)
to (1.1, -3); 
\draw [fill=\fillA, draw=none, yscale=-1]
(-0.9, -3)
to [out=up, in=left](0,-2.1)
to (0, -1.9)
to [out=left, in=down](-0.9, -1)
to (-1.1, -1)
to (-1.1, -3);
\draw [fill=\fillB, draw=none, yscale=-1]
 (0.9,-3)
to (0.9, -3)
to [out=up, in=right](0,-2.1)
to (0, -1.9)
to [out=right, in=down](0.9, -1)
to (1.1, -1)
to (1.1, -3);
\draw (0.9,-3)
to [out=up, in=right](0,-2.1)
to (0, -1.9)
to [out=right, in=down](0.9, -1)
to [out=up, in=down](-0.5, 0)
to [out=up, in=down](0.9, 1)
to [out=up, in=right](0, 1.9)
to (0,2.1)
to [out=right, in=down](0.9, 3);
\draw (-0.9,-3)
to [out=up, in=left](0,-2.1)
to (0, -1.9)
to [out=left, in=down](-0.9, -1)
to [out=up, in=down](0.5, 0)
to [out=up, in=down](-0.9, 1)
to [out=up, in=left](0, 1.9)
to (0,2.1)
to [out=left, in=down](-0.9, 3);
\draw (1.1, -3) to (1.1, 3);
\draw (-1.1, -3) to (-1.1, 3);
\node [draw, fill=white, fill opacity=1,shape=circle, rotate=-100, xscale=-1]
 at (0, 2) {$\delta$};
 \node [draw, fill=white, fill opacity=1, shape=circle, rotate=-10, yscale=-1]
 at (-1, 1) {$\delta$};
  \node [draw, fill=white, fill opacity=1, shape=circle, rotate=10]
 at (1, 1) {$\delta$}; 
 \node [draw, fill=white, fill opacity=1, shape=circle, rotate=-80]
 at (0, -2) {$\delta$};
 \node [draw, fill=white, fill opacity=1, shape=circle, rotate=10]
 at (-1, -1) {$\delta$};
  \node [draw, fill=white, fill opacity=1, shape=circle, rotate=-10, yscale=-1]
 at (1, -1) {$\delta$}; 
\end{tikzpicture}
\end{aligned}
\ws=\ws
\begin{aligned}
\begin{tikzpicture} [thick, scale=0.9, fill opacity=0.5]
\draw [fill=\fillA, draw=none, fill opacity=0.5](-0.9, -3)
to [out=up, in=left](0,-2.1)
to (0, -1.9)
to [out=left, in=down](-0.9, -1)
to (-1.1, -1)
to (-1.1, -3);
\draw [fill=\fillA, draw=none]
(1.1, 1)
to (1.1, -1)
to (0.9, -1)
to [out=up, in=down](0.9, 1);
\draw [fill=\fillB, draw=none]
(-1.1, 1)
to (-1.1, -1)
to (-0.9, -1)
to [out=up, in=down](-0.9, 1);
\draw [fill=\fillB, draw=none](0.9, -3)
to [out=up, in=right](0,-2.1)
to (0, -1.9)
to [out=right, in=down](0.9, -1)
to (1.1, -1)
to (1.1, -3); 
\draw [fill=\fillA, draw=none, yscale=-1] (-0.9,-3)
to [out=up, in=left](0,-2.1)
to (0, -1.9)
to [out=left, in=down](-0.9, -1)
to (-1.1, -1)
to (-1.1, -3);
\draw [fill=\fillB, draw=none, yscale=-1]
 (0.9, -3)
to [out=up, in=right](0,-2.1)
to (0, -1.9)
to [out=right, in=down](0.9, -1)
to (1.1, -1)
to (1.1, -3);
\draw (0.9,-3)
to (0.9, -3)
to [out=up, in=right](0,-2.1)
to (0, -1.9)
to [out=right, in=down](0.9, -1)
to (0.9, 1)
to [out=up, in=right](0, 1.9)
to (0,2.1)
to [out=right, in=down](0.9, 3);
\draw (-0.9, -3)
to [out=up, in=left](0,-2.1)
to (0, -1.9)
to [out=left, in=down](-0.9, -1)
to (-0.9, 1)
to [out=up, in=left](0, 1.9)
to (0,2.1)
to [out=left, in=down](-0.9, 3);
\draw (1.1, -3) to (1.1, 3);
\draw (-1.1, -3) to (-1.1, 3);
\node [draw, fill=white, fill opacity=1, shape=circle, rotate=-80]
 at (0, -2) {$\delta$};
 \node [draw, fill=white, fill opacity=1, shape=circle, rotate=10]
 at (-1, -1) {$\delta$};
  \node [draw, fill=white, fill opacity=1, shape=circle, rotate=-10,yscale=-1]
 at (1, -1) {$\delta$}; 
\node [draw, fill=white, fill opacity=1, shape=circle, rotate=-100, xscale=-1]
 at (0, 2) {$\delta$};
 \node [draw, fill=white, fill opacity=1, shape=circle, rotate=-10, yscale=-1]
 at (-1, 1) {$\delta$};
  \node [draw, fill=white, fill opacity=1, shape=circle, rotate=10]
 at (1, 1) {$\delta$}; 
\end{tikzpicture}
\end{aligned}
\ws=\ws
\begin{aligned}
\begin{tikzpicture} [thick, scale=0.9, fill opacity=0.5]
\draw [fill=\fillA, draw=none, fill opacity=0.5](-0.9, -3)
to [out=up, in=left](0,-2.1)
to (0, -1.9)
to [out=left, in=down](-0.9, -1)
to (-0.9, 1)
to [out=up, in=left](0, 1.9)
to (0,2.1)
to [out=left, in=down](-0.9, 3)
to (-0.9, 3)
to (-1.1, 3)
to (-1.1, -1)
to (-1.1, -3);
\draw [fill=\fillB, draw=none](0.9, -3)
to [out=up, in=right](0,-2.1)
to (0, -1.9)
to [out=right, in=down](0.9, -1)
to (0.9, 1)
to [out=up, in=right](0, 1.9)
to (0,2.1)
to [out=right, in=down](0.9, 3)
to (0.9, 3)
to (1.1, 3)
to (1.1, -1)
to (1.1, -3);
\draw (0.9, -3)
to [out=up, in=right](0,-2.1)
to (0, -1.9)
to [out=right, in=down](0.9, -1)
to (0.9, 1)
to [out=up, in=right](0, 1.9)
to (0,2.1)
to [out=right, in=down](0.9, 3);
\draw (-0.9, -3)
to [out=up, in=left](0,-2.1)
to (0, -1.9)
to [out=left, in=down](-0.9, -1)
to (-0.9, 1)
to [out=up, in=left](0, 1.9)
to (0,2.1)
to [out=left, in=down](-0.9, 3);
\draw (1.1, -3) to (1.1, 3);
\draw (-1.1, -3) to (-1.1, 3);
\node [draw, fill=white, fill opacity=1, shape=circle, rotate=-80]
 at (0, -2) {$\delta$}; 
\node [draw, fill=white, fill opacity=1, shape=circle, rotate=-100, xscale=-1]
 at (0, 2) {$\delta$};
\end{tikzpicture}
\end{aligned}
\ws=\ws
\begin{aligned}
\begin{tikzpicture} [thick, scale=0.9, fill opacity=0.5]
\draw [fill=\fillB, draw=none, fill opacity=0.5]
(0.4, -3) 
to(0.6, -3)
to (0.6, 3)
to (0.4, 3);
\draw [fill=\fillA, draw=none]
(-0.4, -3) 
to(-0.6, -3)
to (-0.6, 3)
to (-0.4, 3);
\draw (0.4, -3) to (0.4, 3);
\draw (0.6, -3) to (0.6, 3);
\draw (-0.4, -3) to (-0.4, 3);
\draw (-0.6, -3) to (-0.6, 3);
\end{tikzpicture}
\end{aligned}
\ifbool{arxiv}{\end{equation}}{\end{align}\end{figure*}}
The second can be proven similarly.
\end{proof}

\subsection{Complementary structures in \cat{GpdAct}}
\label{sec:concretecomplementarity}

\begin{lemma}
\label{lem:compingpdact}
Any group $\cat G$ gives rise to a complementary structure in \cat{GpdAct}.
\end{lemma}
\begin{proof}
Let $\cat{G}$ be a group seen as a one-object groupoid. Let $\cat{|G|}$ be the discrete groupoid whose objects are given by $\Mor(\cat G)$, and whose only morphisms are identity morphisms. Then since \cat G and \cat{|G|} have isomorphic sets of morphisms, by Lemma~\ref{lem:microstates} we can find a unitary 2\-morphism $\delta$ of the following type:

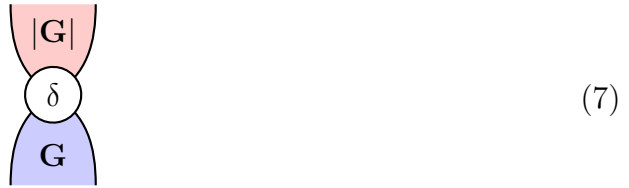
\begin{equation}
\begin{aligned}
\begin{tikzpicture} [thick, scale=0.8, fill opacity=0.5]
\draw [fill=\fillB, draw=none, fill opacity=0.5] (0.7,0)
to [out=down, in=\neangle](0,-1.5)
to [out=\nwangle, in=down](-0.7, 0);
\draw  (0.7,0)
to [out=down, in=\neangle](0,-1.5)
to [out=\nwangle, in=down](-0.7, 0);
\draw [fill=\fillA, draw=none] (0.7,-3)
to [out=up, in=\seangle](0,-1.5)
to [out=\swangle, in=up](-0.7, -3);
\draw (0.7,-3)
to [out=up, in=\seangle](0,-1.5)
to [out=\swangle, in=up](-0.7, -3);
\node [draw, fill=white, fill opacity=1, rotate=10, shape=circle]
 at (0, -1.5) {$\delta$}; 
\node [fill opacity=1] at (0, -0.5) {$\cat{|G|}$}; 
\node [fill opacity=1] at (0, -2.5) {$\cat{G}$}; 
\end{tikzpicture}
\end{aligned}
\end{equation}
Interpreted as a span $\delta$ is a bijection between morphisms of \cat{G} and \cat{|G|}, and $\delta^{\dagger}$ is its converse, hence also a bijection. It follows that $\delta\circ\delta^{\dagger}=\delta^{\dagger}\circ\delta= \id$, which establishes the first property of Definition~\ref{def:complementarystructure}.

To establish the second property we pick arbitrary morphisms $g, g' \in \Mor(\cat{G})$, and verify the unitarity of the partial transpose of $\delta$ on the pair $(g, \delta(g'))$. Since $\delta$ is a bijection, this covers the entire domain of both sides of the equation. Each segment between the horizontal dashed lines is labelled with the current form of the initial pair $(g,\delta(g'))$:
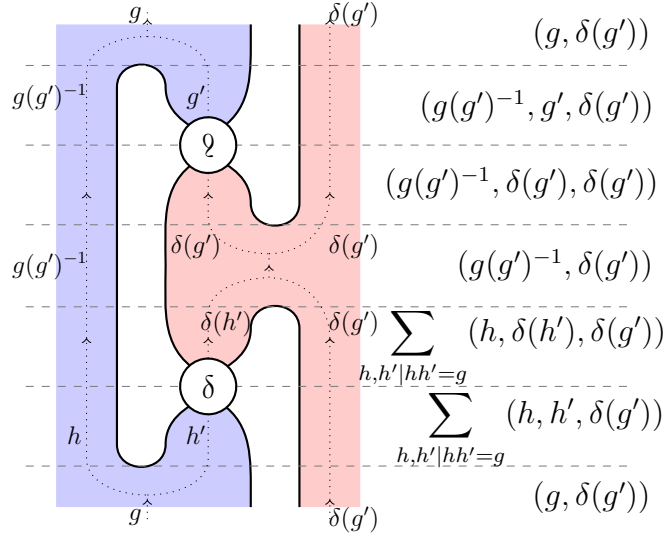
\begin{align*}
\begin{aligned}
\begin{tikzpicture} [scale=0.8, yscale=0.8, fill opacity=0.5]
\draw [fill=\fillB, draw=none, fill opacity=0.5] (2.5, 0)
to (2.5, 5)
to (1.5,5)
to (1.5, 1.3)
to [out=down, in=down, looseness=2](0.7, 1.3)
to [out=up, in=\seangle](0, 2.3)
to [out=\swangle, in=up](-0.7,0)
to [out=down, in=\nwangle](-0,-2.3)
to [out=\neangle, in=down](0.7, -1.3)
to [out=up, in=up, looseness=2](1.5,-1.3)
to (1.5, -5)
to (2.5, -5);
\draw [thick](1.5,5)
to (1.5, 1.3)
to [out=down, in=down, looseness=2](0.7, 1.3)
to [out=up, in=\seangle](0, 2.3)
to  [out=\swangle, in=up](-0.7, 0)
to [out=down, in=\nwangle](-0,-2.3)
to [out=\neangle, in=down](0.7, -1.3)
to [out=up, in=up, looseness=2](1.5,-1.3)
to (1.5, -5);
\draw [fill=\fillA, draw=none] (-2.5,0)
to (-2.5, 5)
to (0.7, 5)
to (0.7, 4.5)
to [out=down, in=\neangle](0, 2.7)
to [out=\nwangle, in=down](-0.7, 3.7)
to [out=up, in=up, looseness=2](-1.5, 3.7)
to (-1.5, 0)
to (-1.5, -3.7)
to [out=down, in=down, looseness=2](-0.7, -3.7)
to [out=up, in=\swangle](0,-2.7)
to [out=\seangle, in=up](0.7, -4.5)
to (0.7, -5)
to (-2.5,-5);
\draw [thick](0.7, 5)
to (0.7, 4.5)
to [out=down, in=\neangle](0, 2.7)
to [out=\nwangle, in=down](-0.7, 3.7)
to [out=up, in=up, looseness=2](-1.5, 3.7)
to(-1.5, 0)
to (-1.5, -3.7)
to [out=down, in=down, looseness=2](-0.7, -3.7)
to [out=up, in=\swangle](0,-2.7)
to [out=\seangle, in=up](0.7, -4.5)
to (0.7, -5); 
\draw [dashed, gray] (-3, -4.15) to (7,-4.15); 
\draw [dashed, gray] (-3, -2.5) to (7,-2.5);
\draw [dashed, gray] (-3, -0.85) to (7,-0.85);
\draw [dashed, gray] (-3, 4.15) to (7,4.15); 
\draw [dashed, gray] (-3, 2.5) to (7,2.5);
\draw [dashed, gray] (-3, 0.85) to (7,0.85);
\draw [dotted, ->] (-1, -5.25) to (-1, -4.85);
\draw [dotted, ->] (-1, -4.85) to (-1, -4.75)
to [out=left, in=down](-2,-4)
to (-2, -1.5);
\draw [dotted, ->] (-2, -1.5) to (-2,1.5);
\draw [dotted, ->](-2,1.5)
to (-2,4)
to [out=up, in =left](-1,4.75)
to (-1, 5.1);
\draw [dotted] (-1, 5.1) to (-1, 5.25);
\draw [dotted, ->] (-1, -4.75)
to [out=right, in=down](0,-4)
to (0, -1.5);
\draw [dotted, ->](0, -1.5)
to (0,-1)
to [out=up, in=left](1, -0.25)
to (1,0);
\draw [dotted, ->](1,0)
to (1, 0.25)
to [out=left, in=down](0,1)
to (0,1.5);
\draw [dotted](0,1.5)
to (0,4)
to [out=up, in=right](-1, 4.75);
\draw [dotted, ->] (2, -5.25) to (2, -4.85);
\draw [dotted, ->] (2, -4.85)
to (2, -1.5);
\draw [dotted](2, -1.5)
to [out=up, in=right](1, -0.25);
\draw [dotted, ->] (1, 0.25)
to [out=right, in=down](2, 1.5);
\draw [dotted, ->](2, 1.5)
to (2, 5.1);
\draw [dotted] (2,5.1) to (2, 5.25);
\node [draw, fill=white, fill opacity=1, rotate=10, shape=circle, thick]
 at (0, -2.5) {$\delta$}; 
\node [draw, fill=white, fill opacity=1, yscale=-1, rotate=10, shape=circle, thick]
 at (0, 2.5) {$\delta$};
\node [fill opacity=1] at (6.3, -4.8){$(g, \delta(g'))$};
\node [fill opacity=1] at (5.3, -3.4){$\displaystyle \sum_{{h, h' | hh'=g}}(h,h', \delta(g'))$};
\node [fill opacity=1]  at (5, -1.7){$\displaystyle \sum_{h,h' | hh'=g}(h,\delta(h'), \delta(g'))$};
\node [fill opacity=1] at (5.7, 0){$(g(g')^{-1}, \delta(g'))$};
\node [fill opacity=1] at (5.15, 1.7){$(g(g')^{-1}, \delta(g'),  \delta(g'))$};
\node [fill opacity=1] at (5.4, 3.3){$(g(g')^{-1}, g',  \delta(g'))$};
\node [fill opacity=1] at (6.3, 4.8){$(g,  \delta(g'))$};
\node [scale=0.8, fill opacity=1]at (-1.2, -5.2){$g$};
\node [scale=0.8, fill opacity=1]at (-1.2, 5.2){$g$};
\node [scale=0.8, fill opacity=1]at (2.4, -5.25){$\delta(g')$};
\node [scale=0.8, fill opacity=1]at (-2.2, -3.5){$h$};
\node [scale=0.8, fill opacity=1]at (-0.2, -3.5){$h'$};
\node [scale=0.8, fill opacity=1]at (-0.2, 3.5){$g'$};
\node [scale=0.8, fill opacity=1] at (-2.6, 3.5){$g(g')^{-1}$};
\node [scale=0.8, fill opacity=1] at (-2.6, 0){$g(g')^{-1}$};
\node [scale=0.8, fill opacity=1] at (2.4, 5.2){$\delta(g')$};
\node [scale=0.8, fill opacity=1] at (2.4, 0.4){$\delta(g')$};
\node [scale=0.8, fill opacity=1] at (2.4, -1.2){$\delta(g')$};
\node [scale=0.8, fill opacity=1] at (0.3, -1.2){$\delta(h')$};
\node [scale=0.8, fill opacity=1] at (-0.2, 0.4){$\delta(g')$};
\end{tikzpicture}
\end{aligned}
\end{align*}
The crucial step is when we extract the element of the sum $\sum_{hh'=g}(h,\delta(h'), \delta(g'))$ for which $\delta(h')=\delta(g')$. As $\delta$ is a bijection, this implies $h'=g'$. Hence $h=g(h')^{-1}=g(g')^{-1}$, and the whole procedure equals the identity as required.
\end{proof}

\subsection{Communication structures in \cat{GpdAct}}
\label{sec:concreteencryption}

\noindent
By applying Theorem~\ref{thm:commfromcomp} to the complementary structures described by Lemma~\ref{lem:compingpdact}, we obtain we obtain a communication structure in $\GA$ associated to any group.

It is well-established that communication structures give syntactic descriptions of encrypted communication procedures, for which $\lambda$ is the encryption operation, whose first input is the plaintext and whose second input is the secret key. To see how this operates in our case, we consider the action of the encryption procedure $\lambda$ on a pair of elements $(g,\delta(g'))$:
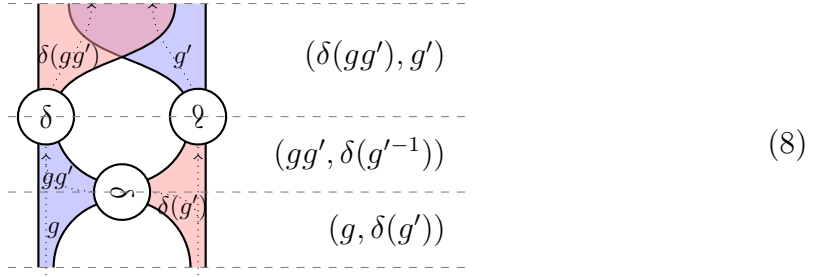
\begin{equation}
\begin{aligned}
\begin{tikzpicture} [scale=1, fill opacity=0.5]
\draw [fill=\fillA, draw=none, fill opacity=0.5]
(-0.9, -4.5)
to [out=up, in=left](0,-3.6)
to (0, -3.4)
to [out=left, in=down](-0.9, -2.5)
to (-1.1, -2.5)
to (-1.1, -4.5);
\draw (-0.9, -4.5) [thick]
to [out=up, in=left](0,-3.6)
to (0, -3.4)
to [out=left, in=down](-0.9, -2.5)
to (-1.1, -2.5)
to (-1.1, -4.5);
\draw [fill=\fillA, draw=none]
(1.1, -1.0)
to (1.1, -2.5)
to (0.9, -2.5)
to [out=up, in=down](-0.7, -1.0);
\draw (1.1, -1)  [thick]
to (1.1, -2.5)
to (0.9, -2.5)
to [out=up, in=down](-0.7, -1);
\draw [fill=\fillB, draw=none]
(-1.1, -1)
to (-1.1, -2.5)
to (-0.9, -2.5)
to [out=up, in=down](0.7, -1);
\draw (-1.1, -1)  [thick]
to (-1.1, -2.5)
to (-0.9, -2.5)
to [out=up, in=down](0.7, -1);
\draw [fill=\fillB, draw=none]
(0.9, -4.5)
to [out=up, in=right](0,-3.6)
to (0, -3.4)
to [out=right, in=down](0.9, -2.5)
to (1.1, -2.5)
to (1.1, -4.5);
\draw (0.9, -4.5)  [thick]
to [out=up, in=right](0,-3.6)
to (0, -3.4)
to [out=right, in=down](0.9, -2.5)
to (1.1, -2.5)
to (1.1, -4.5);
\draw [dotted, ->] (-1, -4.6) to (-1, -3);
\draw [dotted, ->] (-1, -3) to (-1, -2.5) to [out=up, in=down](-0.4, -1.0);
\draw [dotted, ->] (1, -4.6) to (1, -3);
\draw [dotted, ->] (1, -3) to (1, -2.5) to [out=up, in=down](0.4, -1);
\draw [dotted] (1, -3.9) to [out=up, in=right](0, -3.5) to [out=left, in=down](-1, -3.1);
\node [circlelabel, fill opacity=1, rotate=-80] at (0, -3.5) {$\delta$};
\node [circlelabel, fill opacity=1, rotate=10] at (-1, -2.5) {$\delta$};
\node [circlelabel, fill opacity=1, yscale=-1, rotate=10] at (1, -2.5) {$\delta$};
\node [anchor=east, fill opacity=1] at (4.4, -3) {$(gg', \delta(g' {}^{-1}))$};
\node [anchor=east, fill opacity=1] at (4.4, -1.7) {$(\delta(gg'), g')$};
\node [anchor=east, fill opacity=1] at (4.4, -4) {$(g,\delta(g'))$};
\node [scale=0.8, fill opacity=1] at (0.8, -3.7) {$\delta( g')$};
\node [scale=0.8, fill opacity=1] at (0.8, -1.7) {$g'$};
\node [scale=0.8, fill opacity=1] at (-0.9, -4) {$g$};
\node [scale=0.8, fill opacity=1] at (-0.85, -3.3) {$gg'$};
\node [scale=0.8, fill opacity=1] at (-0.7, -1.7) {$\delta(gg')$};
\draw [dashed, gray] (-1.5, -4.5) to (4.6, -4.5);
\draw [dashed, gray] (-1.5, -3.5) to (4.6, -3.5);
\draw [dashed, gray] (-1.5, -2.5) to (4.6, -2.5);
\draw [dashed, gray] (-1.5, -1) to (4.6, -1);
\end{tikzpicture}
\end{aligned}
\end{equation}
The partial transpose of $\delta$ has the same effect on $(g, \delta(g'))$ as in \ref{lem:compingpdact}. The output $(\delta(gg'), g')$ is a morphism in the product $\cat {|G|} \times \cat G$, whose factors correspond to the red and blue regions respectively at the top of the diagram.

The information about the logical state $\delta(gg')$ is entirely contained in the first factor. Since $g'$ represents the secret key, which could take any value, the value $\delta(gg')$ reveals no information about the plaintext $g$ to an adversary who is ignorant of $g'$. For this reason, $\delta(gg')$ is indeed an appropriate ciphertext for the encryption process.

The second factor of the output of the encryption process is the morphism $g'$ in \cat G. Since this groupoid has only one object, this is pure microstate information, and we can interpret it as the heat output of the process. The value of this microstate is precisely the secret key; by the analysis of Section~\ref{sec:phenomenology}, this will be unavailable to any adversary who cannot control the entire system $\BG \circ \GB$, preserving the security of the protocol in most circumstances.

\section{Quantization}
\label{sec:quantum}

\subsection{Introduction}

\noindent
In physics, it is currently accepted that quantum theory is the fundamental theory of physical processes, at least when gravitation is excluded. For this reason, if we have a categorical semantics that purports to describe the physical context of classical computation, it is interesting to ask whether it has an interpretation in a quantum semantics. We call such an interpretation \emph{quantization}, after the standard process in quantum theory where a classical description of a system is converted into a quantum description.

We show in Section~\ref{sec:definitionQ} that such a quantization procedure can indeed be defined, in the form of a pseudofunctor $$\mbox{$Q : \GA \to \cat{2Hilb}$},$$ where \cat{2Hilb} is a bicategory that is already known to be suitable for the description of quantum computational processes~\cite{b97-hda2, v12-hsqp}.
Applying $Q$ to the model of encrypted communication given in Section~\ref{sec:encryption} yields exactly quantum teleportation, a result we prove in Section~\ref{sec:teleportation}. This is the first demonstration in the literature that that quantum teleportation has a functorial classical analogue, an attractive result which we discuss in greater length in Section~\ref{sec:otherwork}.

Having this explicit mapping from encrypted communication to quantum teleportation in hand, we can examine how certain thermodynamic and quantum phenomena are related. We show in Section~\ref{sec:decoherence} that the thermodynamic model exhibits a classical version of decoherence, which under $Q$ given back quantum decoherence. In Section~\ref{sec:densecoding} we show that quantum dense coding also has a classical thermodynamic representation, which is its preimage under $Q$.

\subsection{Definition of $Q$}
\label{sec:definitionQ}

\begin{defn}
The quantization pseudofunctor$$Q : \GA \to \cat{2Hilb}$$ is defined as follows:
\begin{itemize}
\item On objects, for a groupoid \cat G, we have $Q(\cat G) := \Rep(\cat G)$
\item On morphisms, for a functor $S: \cat H ^\op \times \cat G \to \cat{Set}$, we have $Q(S) := \Lin \circ S$, where $\Lin:\cat{Set} \to \cat{Hilb}$ is the linearization functor
\item On 2-morphisms, for a natural transformation \mbox{$\alpha: S \times T \to \N$}, we define \mbox{$Q(\alpha):Q(S) \Rightarrow Q(T)$} as the natural transformation given at stage $(H,G)$ by  $Q(\alpha) _{(H,G)}(s) := \sum_t \alpha(s,t) _{(H,G)} \, t$
\end{itemize}
\end{defn}

\noindent
For a morphism $S:\cat G \to \cat H$ the prescription above yields a representation of $\cat H ^\op \times \cat G$; by the compact structure of \cat{2Hilb}~\cite{b97-hda2} this is canonically equivalent to a linear functor $\Rep(\cat G) \to \Rep(\cat H)$.

Several lemmas are required to establish that this definition is valid.

\begin{lemma}
For a 2\-morphism $\alpha: S \times T \to \N$ in \GA, the family of maps $Q(\alpha)$ satisfies the naturality condition.
\end{lemma}
\begin{proof}
The naturality condition amounts to the requirement that the linear maps are intertwiners for the group actions. Suppose we have a morphism $(h,g):(H',G) \to (H,G')$ in $\cat H ^\op \times \cat G$. Then we test the intertwiner property as follows:
\begin{align*}
Q(\alpha)_{(H,G')}(h \cdot s \cdot g)
&= \textstyle \sum_{t \in T(H,G')} \alpha(h \cdot s \cdot g, t) \, t
\\
&= \textstyle \sum_{t \in T(H,G')} \alpha(s,t)\, t
\\
&= \textstyle \sum_{t \in T(H,G')} \alpha(s,h \cdot t \cdot g)\, t
\\
&= \textstyle \sum _{t' \in T (H',G)} \alpha(s,t) \, (h \cdot t' \cdot g)
\\
&= h \cdot \big(Q(\alpha) _{(H',G)}(s) \big) \cdot g
\end{align*}
This completes the proof.
\end{proof}

\begin{lemma}
 $Q$ preserves vertical composition.
\end{lemma}

\begin{lemma}
For morphisms $S: \cat G \to \cat H$ and $T : \cat H \to \cat J$ in \GA, we have
$$Q(T \circ S) \simeq Q(T) \circ Q(S).$$
\end{lemma}
\begin{proof}
Our goal is to show that the linearized profunctor composition operation gives the same effect as composition in~\cat{2Hilb}. Since \cat{2Hilb} is compact, the morphisms $Q(S)$ and $Q(T)$ can be curried to give
\begin{align*}
Q(S)' &: \cat{Hilb} \to Q(\cat H)^\op \boxtimes Q(\cat G)
\\
Q(T)' &: \cat{Hilb} \to Q(\cat J) ^\op \boxtimes Q(\cat H),
\end{align*}
and the composed morphism $(Q(T) \circ Q(S))'$ can be given in this perspective as
$$(\id _{Q(\cat J) ^\op} \boxtimes \Hom _{Q(\cat H)}(-,-) \boxtimes \id _{Q(\cat G)}) \circ (Q(T)' \boxtimes Q(S)'),$$
where for brevity we omit the weak structure of the tensor product.

Fix objects $G \in \cat G$, $H \in \cat H$ and $J \in \cat J$, and choose $t \in T(J,H)$ and $s \in S(H,G)$. Then $(t,s)$ is an element of the profunctor composite $T \circ S$ at stage $(J,G)$. For any morphism $h:H \to H'$ in \cat H, the pair $(t,s)$ is equivalent to any element of the form $(T(J,h)(t),S(h ^{-1}, G)(s))$, which we write using compressed notation as $(t \cdot  h, h^{-1} \cdot s)$, suppressing the objects $J$ and $G$ as they are presumed fixed. Since we are working with groupoids, this gives the entire equivalence relation of profunctor composition.

It follows that we must demonstrate, for all $s$ and $t$, a particular isomorphism of vector spaces. The first vector space is the free complex vector space of equivalence classes of pairs of elements in the orbits of $t$ and $s$, i.e. pairs of the form $(t \cdot h, h' \cdot s)$, which are all equivalent to pairs of the form $(t \cdot h, s)$. The second vector space is the space of intertwiners between the actions of $\cat H$ on the free complex vector space on the orbits of $t$ and $s$; such an intertwiner is a linear map \mbox{$L: \Orbit _{\cat H} (t) \to \Orbit_{\cat H} (s)$}, where we allow these sets to stand for their free complex vector spaces over them, satisfying the condition $L(t \cdot h) = h \cdot L(t)$ for all morphisms $h \in \cat H$.

Given an intertwiner $L$, we map it to the following linear combination $\pi(L)$ of pairs of elements:
\begin{align*}
\pi(L) &:= (t,L(t))
\intertext{Conversely, given a pair of elements $(t \cdot h, s)$, we map it to the intertwiner $\sigma(t \cdot h, s)$ defined as follows:}
\sigma (t \cdot h, s)(t \cdot h') 
&:=
{|\Stab_{\cat H} (t)|} ^{-1}
\textstyle\sum _{h_0 \in \Stab _{\cat H} (t)} h' h_0 h \cdot s
\end{align*}
We must show that  $\sigma(t \cdot h, s)$ is well-defined. Suppose \mbox{$h' \in \Stab _{\cat H} (t)$}. Then we can shift the summation variable \mbox{$h_0 \leadsto h' h_0$}, and we see that $\sigma (t \cdot h, s)$ takes the same value at $t$ and $t \cdot h'$ as required.

We must also show that $\sigma$ itself is well-defined. Suppose $(t \cdot h_1, s) \sim (t \cdot h_2, s)$; then $\sigma$ must give the same intertwiner on each of these. For these pairs to be equivalent is the same condition as saying $(t,h_1 \cdot s) \sim (t, h_2 \cdot s)$. This means there exists some $h''$ with $t \cdot h'' = t$, and $h'' {}^{-1} h_1 \cdot s = h_2 \cdot s$. The first of these equations simply says $h'' \in \Stab _{\cat H} (t)$. We then apply our definition of $\sigma$ and perform the following computation:
\begin{align*}
\sigma (t \cdot h_2, s)(t \cdot h')
\hspace{-2cm}
\\
&=
\textstyle {|\Stab_{\cat H} (t)|} ^{-1}
\sum _{h_0 \in \Stab _{\cat H} (t)} h' h_0 h_2 \cdot s
\\
&=
\textstyle {|\Stab_{\cat H} (t)|} ^{-1}
\sum _{h_0 \in \Stab _{\cat H} (t)} h' h_0 h'' {}^{-1} h_1 \cdot s
\\
&=
\textstyle {|\Stab_{\cat H} (t)|} ^{-1}
\sum _{h' _0 \in \Stab _{\cat H} (t)} h' h'_0 h_1 \cdot s
\\
&= \sigma (t \cdot h_1, s)(t \cdot h')
\end{align*}
At the second equality we use the identity obtained immediately above, and at the third equality we shift $h_0$ by $h''$. So $\sigma$ is well-defined as required.

Finally we demonstrate that the constructions $\sigma$ and $\pi$ are inverse. Given an arbitrary intertwiner $L$, we first perform the following calculation:
\begin{align*}
\sigma (\pi (L))(t \cdot h)
&=
{|\Stab_{\cat H} (t)|} ^{-1}
\textstyle\sum _{h_0 \in \Stab _{\cat H} (t)} h h_0 \cdot L(t)
\\
&=
{|\Stab_{\cat H} (t)|} ^{-1}
\textstyle\sum _{h_0 \in \Stab _{\cat H} (t)}  L(t \cdot h_0 h)
\\
&=
{|\Stab_{\cat H} (t)|} ^{-1}
\textstyle\sum _{h_0 \in \Stab _{\cat H} (t)} L(t \cdot h)
\\
&= L(t \cdot h)
\end{align*}
For the other direction we apply our constructions to an arbitrary pair $(t \cdot h, s)$:
\begin{align*}
\pi(\sigma(t \cdot h, s))
&= \left( t, {|\Stab_{\cat H} (t)|} ^{-1}
\textstyle\sum _{h_0 \in \Stab _{\cat H} (t)} h_0 h \cdot s \right) 
\\
&= | \Stab_{\cat{H}}(t) | ^{-1} \textstyle\sum _{h_0 \in \Stab _{\cat H} (t)} ( t, h_0 h \cdot s)
\\
&\sim | \Stab_{\cat{H}}(t) | ^{-1} \textstyle\sum _{h_0 \in \Stab _{\cat H} (t)} ( t \cdot h_0 h, s)
\\
&= | \Stab_{\cat{H}}(t) | ^{-1} \textstyle\sum _{h_0 \in \Stab _{\cat H} (t)} ( t \cdot h, s)
\\
&= (t \cdot h, s)
\end{align*}
This completes the proof.
\end{proof}

It is beyond the scope of this paper to demonstrate that $Q$ preserves all the structure that is present, such as the symmetric monoidal structure. We leave this as a conjecture.
\begin{conjecture} The construction $Q$ can be promoted into a symmetric monoidal pseudofunctor.
\end{conjecture}

\subsection{Encrypted communication and quantum teleportation}
\label{sec:teleportation}

\noindent
By the results of Section~\ref{sec:definitionQ}, the procedure $Q$ will translate the communication structures described in Section~\ref{sec:concreteencryption} from $\GA$ into \cat{2Hilb}. It was shown in~\cite{v12-hqt, v12-hsqp} that communication structures in~\cat{2Hilb} correspond exactly to implementations of quantum teleportation, so we can immediately conclude that $Q$ gives a functorial transformation of encrypted communication into quantum teleportation.

As discussed in Sections~\ref{sec:overview} and~\ref{sec:otherwork}, this goes significantly beyond existing results: for the first time, it demonstrates a `toy model' of quantum theory which  gives back the correct quantum procedures in a functorial way.

\subsection{Classical decoherence}
\label{sec:decoherence}

\noindent
Quantum decoherence is the process by which quantum systems become entangled with environmental degrees of freedom, often undesirably, causing entanglement between the parts of our quantum system to degenerate into classical correlation~\cite{decoherence}. This phenomenon is at the heart of the measurement problem~\cite{s05-decoherence}---one of the most profound unsolved problems in the foundations of physics---and so any insight into its mathematical structure is interesting to pursue.

A direct analogue of decoherence is present in our model. Suppose that we erase some logical information using a process of the form $\delta ^\dag$, as described in Section~\ref{sec:concretecomplementarity}. From a microscopic perspective this is completely reversible in principle, as long as the entire target system $\cat G$ is kept within our control; the logical information is now encoded in the microscopic state of \cat G. However, if the system \cat G interacts with its environment in any way, even reversibly, then by the results of Section~\ref{sec:phenomenology} this will in general perturb the microstate and render our original information irretrievable.

Applying our quantization process $Q$ to this thermodynamic scenario, and by applying the results of~\cite{v12-hqt, v12-hsqp},  we obtain exactly a description of decoherence and its effect in quantum theory. The sequence of events is as follows: we begin by measuring a quantum system. Just as in the thermodynamic case, from a microscopic perspective, this is completely reversible in principle, since the measurement must be enacted by a unitary operation. However, if the target system interacts with the environment in any way, then the quantum system will decohere into a classical mixture, and our original quantum state will be lost.

\subsection{Classical dense coding}
\label{sec:densecoding}

\noindent
A communication structure in a bicategory provides a solution to the following equation~\cite{v12-hqt, v12-hsqp}:
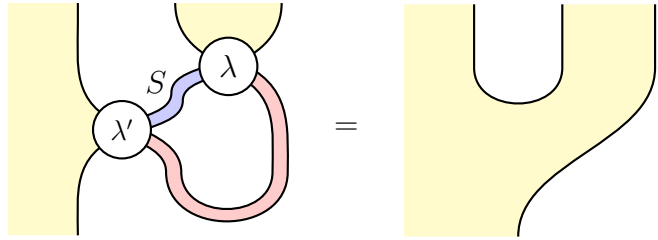
\begin{equation*}
\begin{aligned}
\begin{tikzpicture} [thick, yscale=0.7, fill opacity=0.5]
\draw [fill=\fillB, fill opacity=0.5, draw=none]  (0.78,-1)
to [out=up, in=\seangle](0.1,-0.1)
to (0, -0.3)
to [out=\seangle, in=up](0.58, -1)
to  [out=down, in=down, looseness=2](2.18,-1)
to [out=up, in=\seangle](1.5,1.1)
to (1.4, 0.9)
to [out=\seangle, in=up] (1.98, -1)
to [out=down, in=down, looseness=2](0.78, -1);
\draw   (0.78,-1)
to [out=up, in=\seangle](0.1,-0.1)
to (0, -0.3)
to [out=\seangle, in=up](0.58, -1)
to  [out=down, in=down, looseness=2](2.18,-1)
to [out=up, in=\seangle](1.5,1.1)
to (1.4, 0.9)
to [out=\seangle, in=up] (1.98, -1)
to [out=down, in=down, looseness=2](0.78, -1);
\draw [fill=\fillD, draw=none] (-1.5,-2.2)
to (-1.5,2.2)
to (-0.58,2.2)
to (-0.58,1.8)
to [out=down, in=\nwangle](0,-0.1)
to (0, -0.3)
to [out=\swangle, in=up](-0.58, -2.2);
\draw (-0.58,2.2)
to (-0.58,1.8)
to [out=down, in=\nwangle](0,-0.1)
to (0, -0.3)
to [out=\swangle, in=up](-0.58, -2.2);
\draw [fill=\fillD, draw=none] (0.7, 2.2) 
to [out=down, in=\nwangle](1.4,1)
to [out=\neangle, in=down](2.1,2.2);
\draw (0.7, 2.2) 
to [out=down, in=\nwangle](1.4,1)
to [out=\neangle, in=down](2.1,2.2);
\draw [fill=\fillA, draw=none](0,-0.1)
to [out=\neangle, in=down] (0.65, 0.5)
to [out=up, in=\swangle](1.3, 1.1)
to (1.4, 0.9)
to [out=\swangle, in=up](0.8, 0.35)
to [out=down, in=\neangle](0.1, -0.3);
\draw (0,-0.1)
to [out=\neangle, in=down] (0.65, 0.5)
to [out=up, in=\swangle](1.3, 1.1)
to (1.4, 0.9)
to [out=\swangle, in=up](0.8, 0.35)
to [out=down, in=\neangle](0.1, -0.3);
\node [draw, fill=white, fill opacity=1, shape=circle, fill opacity=1, scale=0.9]
 at (0, -0.2) {$\lambda'$};
\node [draw, fill=white, fill opacity=1, shape=circle, fill opacity=1]
 at (1.4, 1) {$\lambda$};
\node [fill opacity=1]
 at (0.45, 0.7) {$S$};
\end{tikzpicture}
\end{aligned}
\quad=\quad
\begin{aligned}
\begin{tikzpicture} [thick, yscale=0.7, fill opacity=0.5]
\draw [fill=\fillD, draw=none] (-1.5,-2.2)
to (-1.5,2.2)
to (-0.58,2.2)
to (-0.58,1.0)
to [out=down, in=down, looseness=2](0.58,1.0)
to (0.58,2.2)
to (1.8, 2.2)
to (1.8,1)
to [out=down, in=up] (0,-2.2);
\draw (-0.58,2.2)
to (-0.58,1.0)
to [out=down, in=down, looseness=2](0.58,1.0)
to (0.58,2.2);
\draw (1.8, 2.2)
to (1.8,1)
to [out=down, in=up] (0,-2.2);
\end{tikzpicture}
\end{aligned}
\end{equation*}
The left-hand side represents a sequence of processes whose overall result is splitting the system into two, as described by the right-hand side. An easy assumption is that the channel marked $S$ must therefore have sufficient capacity to describe the state of one copy of the system. Yet in fact $\lambda$ here is invertible, and $S$ has only the square root of the apparent necessary capacity.

A realization of this abstract procedure over a quantum semantics gives quantum dense coding~\cite{v12-hqt}, which in a sense allows 2 classical bits to be transferred along a quantum channel whose classical capacity is only 1 bit. This is recognized as one the most remarkable quantum informatic procedures~\cite{bw92}. Yet it follows from our results that it has a classical thermodynamic analogue, and that taking the image of this under $Q$ produces the well-known quantum version.

The thermodynamic analogue is given by the 2\-morphism $\lambda$ constructed in Section~\ref{sec:concreteencryption}, and its partial transpose $\lambda'$. By tracing through the action of $\lambda'$, it can be seen that the procedure succeeds thanks to perturbation of the microstate of the left-hand extended system by the free system $S$. This seems remarkable given that no such action is obviously present in the standard quantum solution, even though it is obtained from the thermodynamic solution functorially.

\ifbool{arxiv}{}{\fontsize{9.7}{11.64}\selectfont}
\newlength\mybibindent
\setlength\mybibindent{-17pt}
\makeatletter
\renewenvironment{thebibliography}[1]
     {\section*{\bibname}%
      \@mkboth{\MakeUppercase\bibname}{\MakeUppercase\bibname}%
      \list{\@biblabel{\@arabic\c@enumiv}}%
           {\settowidth\labelwidth{\@biblabel{1}}
            \leftmargin\labelwidth
            \advance\leftmargin\dimexpr\labelsep+\mybibindent\relax\itemindent-\mybibindent
            \@openbib@code
            \usecounter{enumiv}%
            \let\p@enumiv\@empty
            \renewcommand\theenumiv{\@arabic\c@enumiv}}%
      \sloppy
      \clubpenalty4000
      \@clubpenalty \clubpenalty
      \widowpenalty4000%
      \sfcode`\.\@m}
     {\def\@noitemerr
       {\@latex@warning{Empty `thebibliography' environment}}%
      \endlist}
\makeatother

\bibliographystyle{plain}
\bibliography{references}

\end{document}